\documentclass[a4paper, 11pt]{article}
\usepackage[usenames]{xcolor}
\usepackage{xifthen}
\usepackage{todonotes}
\def\fontsettingup{2} 

\usepackage{amsmath, amsthm, amssymb}
\usepackage[T1]{fontenc}
\ifthenelse{\fontsettingup = 1}{ \usepackage{eulerpx, eucal, tgpagella}   }{}
\ifthenelse{\fontsettingup = 2}{ \usepackage{mathpazo, tgpagella} }{}

\usepackage{hyperref}
\hypersetup{
  colorlinks=true,
  linkcolor=blue,
  filecolor=magenta,
  urlcolor=cyan,
}

\usepackage{algorithm2e}
\usepackage{tikz}
\usepackage{cleveref}

\usepackage[a4paper]{geometry}
\newgeometry{
textheight=9in,
textwidth=6.5in,
top=1in,
headheight=14pt,
headsep=25pt,
footskip=30pt
}

\newtheorem{theorem}{Theorem}

\newtheorem*{claim*}{Claim}

\newtheorem{example}[theorem]{Example}

\newtheorem{lemma}[theorem]{Lemma}
\newtheorem{proposition}[theorem]{Proposition}
\newtheorem{corollary}[theorem]{Corollary}
\theoremstyle{definition}

\newtheorem{definition}[theorem]{Definition}
\newtheorem{construction}[theorem]{Construction}
\newtheorem{remark}[theorem]{Remark}
\newtheorem*{remark*}{Remark}

\ifthenelse{\fontsettingup = 1}{
  \def\*#1{\mathbf{#1}} 
  \def\+#1{\mathcal{#1}} 
  \def\-#1{\mathrm{#1}} 
  \def\=#1{\mathbb{#1}} 
  \def\!#1{\mathfrak{#1}} 
}{}

\ifthenelse{\fontsettingup = 2}{
  \def\*#1{\boldsymbol{#1}} 
  \def\+#1{\mathcal{#1}} 
  \def\-#1{\mathrm{#1}} 
  \def\=#1{\mathbb{#1}} 
  \def\!#1{\mathfrak{#1}} 
}{}

\def\oPr{\mathop{\mathbf{Pr}}}
\renewcommand{\Pr}[2][]{ \ifthenelse{\isempty{#1}}
  {\oPr\left[#2\right]}
  {\oPr_{#1}\left[#2\right]} } 

\def\oE{\mathop{\mathbb{E}}}
\newcommand{\E}[2][]{ \ifthenelse{\isempty{#1}}
  {\oE\left[#2\right]}
  {\oE_{#1}\left[#2\right]} }

\usepackage{xparse}

\def\oVar{\mathbf{Var}}
\NewDocumentCommand{\Var}{ O{} O{} m }{
  \ifthenelse{\isempty{#1}} {
    \ifthenelse{\isempty{#2}} {
      \oVar\left[#3\right]
    } {
      \oVar^{#2}\left[#3\right]
    }
  } {
    \ifthenelse{\isempty{#2}} {
      \oVar_{#1}\left[#3\right]
    } {
      \oVar_{#1}^{#2}\left[#3\right]
    }
  }
}

\def\oEnt{\mathbf{Ent}}
\NewDocumentCommand{\Ent}{ O{} O{} m }{
  \ifthenelse{\isempty{#1}} {
    \ifthenelse{\isempty{#2}} {
      \oEnt\left[#3\right]
    } {
      \oEnt^{#2}\left[#3\right]
    }
  } {
    \ifthenelse{\isempty{#2}} {
      \oEnt_{#1}\left[#3\right]
    } {
      \oEnt_{#1}^{#2}\left[#3\right]
    }
  }
}

\newcommand{\DTV}[2]{\-D_{\mathrm{TV}}\left({#1},{#2}\right)}

\renewcommand{\epsilon}{\varepsilon}
\renewcommand{\emptyset}{\varnothing}
\newcommand{\norm}[1]{\left\Vert#1\right\Vert}
\newcommand{\set}[1]{\left\{#1\right\}}
\newcommand{\tuple}[1]{\left(#1\right)}

\newcommand{\tp}{\tuple}
\newcommand{\ol}{\overline}
\newcommand{\wh}{\hat}
\newcommand{\abs}[1]{\left\vert#1\right\vert}

\newcommand{\HSRV}{local Poincar\'e inequality}

\newcommand{\DCP}{decoupling}
\newcommand{\JS}{\mathrm{JS}}

\title{Faster Mixing of the Jerrum-Sinclair Chain}

\def\AND{\hspace{2ex}}
\author{\small Xiaoyu Chen\thanks{State Key Laboratory for Novel Software Technology, New Cornerstone Science Laboratory, Nanjing University, China. Emails: \texttt{\{chenxiaoyu233, zheju, miaotianshun, zhangxy\}@smail.nju.edu.cn, yinyt@nju.edu.cn}},
  \AND Weiming Feng\thanks{School of Computing and Data Science, The University of Hong Kong, China. Email: \texttt{wfeng@hku.hk}},
  \AND Zhe Ju\footnotemark[1],
  \AND Tianshun Miao\footnotemark[1],
  \AND Yitong Yin\footnotemark[1],
  \AND Xinyuan Zhang\footnotemark[1]}

\date{}

\begin{document}

\maketitle

\begin{abstract}
We show that the Jerrum-Sinclair Markov chain on matchings mixes in time $\widetilde{O}(\Delta^2 m)$ on any graph with $n$ vertices, $m$ edges, and maximum degree $\Delta$, for any constant edge weight $\lambda>0$.
%
%
For general graphs with arbitrary, potentially unbounded $\Delta$,
this provides the first improvement over the classic $\widetilde{O}(n^2 m)$ mixing time bound of Jerrum and Sinclair~(1989) and Sinclair~(1992).

To achieve this, we develop a general framework for analyzing mixing times, combining ideas from the classic canonical path method with the ``local-to-global'' approaches recently developed in high-dimensional expanders, introducing key innovations  to both techniques.

  
%
\end{abstract}

\thispagestyle{empty}

\newpage

\tableofcontents

\thispagestyle{empty}

\newpage 

\setcounter{page}{1}

\section{Introduction }\label{sec:introduction}

Markov chain Monte Carlo (MCMC) is a fundamental paradigm for approximate sampling and counting.
The method involves simulating a Markov chain for sufficiently many steps to generate approximate samples from its stationary distribution.
A key challenge in MCMC is analyzing the \emph{mixing time}, 
which quantifies how many steps required for the chain  to reach near-stationarity and produce reliable samples.

A pioneering contribution to MCMC theory is the seminal work of Jerrum and Sinclair~\cite{jerrum1989approximating}, 
which introduced the \emph{canonical path} method for analyzing Markov chain mixing times.
Originally developed for the Jerrum-Sinclair chain in the context of sampling matchings,
this landmark result has since become a standard topic in various textbooks on randomized algorithms and Markov chains \cite{jerrum2003counting,MR10,levin2017markov},
and laid the foundation for numerous advances in approximate counting and sampling.
Key applications include approximating the permanent \cite{jerrum1989approximating,diaconis1991geometric,jerrum2004polynomial,dyer2017switch} and the Ising model \cite{jerrum1993polynomial,martinelli1994two,goldberg2003computational,guo2017random,dyer2021polynomial,feng2023swendsen}; 
sampling and counting combinatorial objects such as matroid bases \cite{feder1992balanced}, knapsack solutions \cite{dyer1993mildly,morris2004random}, contingency tables \cite{cryan2006rapidly,bezakova2007sampling}, and $H$-colorings \cite{dyer2006systematic};
rapid mixing of local Markov chains for generating random regular graphs \cite{kannan1999simple,cooper2005sampling,tomas2006local,cooper2009flip,greenhill2011polynomial,cooper2019flip,erdos2022mixing} and Glauber dynamics on trees \cite{berger2005glauber,lucier2009glauber,goldberg2010mixing, delcourt2020glauber, dyer2021counting,eppstein2023rapid,CCFV25};
as well as many other key developments and applications \cite{diaconis1991geometric, sinclair1992improved,fontes1998spectral,volker2004expansion,hoory2005simple,benjamini2008long,cryan2008random,colin2011cover,mcquillan2013approximating,crosson2016quantum,caputo2016dynamics,huang2016canonical,yang2016computational,cai2019approximability,david2023improved,frieze2023subexponential,olesker2024geometric}.

The Jerrum-Sinclair chain was originally introduced for sampling matchings. 
Let $G=(V,E)$ be an undirected graph, and let $\lambda > 0$ be an edge weight.
The {monomer-dimer model} on $G$  defines a distribution $\mu$, known as the \emph{monomer-dimer distribution}, over all matchings of $G$, given by
\begin{align*}
\forall M \in \+M_G, \quad \mu(M) \propto \lambda^{\abs{M}},
\end{align*}
where $\+M_G$ denotes the set of matchings in the graph $G=(V,E)$.
%
Given a matching $M \in \+M_G$, we say a vertex $v$ is saturated if it is matched by an edge $e$ in $M$, i.e., if $v\in e \in M$.

The Jerrum-Sinclair chain $P_{\-{JS}}$ updates a matching $X_t$ to $X_{t+1}$ according to the following rules:
\begin{enumerate}
\item Select an edge $e = \{u,v\} \in E$ uniformly at random.
\item Propose a ``candidate'' matching $M$ for $X_{t+1}$ according to the following rules:
\begin{enumerate}
\item (\emph{down} transition) If $e\in X_t$, set $M \gets X_t \setminus \{e\}$.
\item (\emph{up} transition) If both endpoints $u$ and $v$ are not saturated in $X_t$, set $M \gets X_t \cup \{e\}$.
\item (\emph{exchange} transition) If one endpoint of $e=\{u,v\}$ is saturated  and the other is not,  say $u$ is saturated by an edge $f$ and $v$ is not saturated, set $M \gets X_t \cup \{e\} \setminus \{f\}$.
\item Otherwise (if both endpoints $u$ and $v$ are saturated but $e \not \in X_t$), set $M \gets X_t$.
\end{enumerate}
\item Accept the proposal by setting $X_{t+1} \gets M$ with probability $\min \set{1,\frac{\mu(M)}{\mu(X_t)}}$; reject the proposal and set $X_{t+1} \gets X_t$ with the remaining probability.
\end{enumerate}

Additionally, consider the $1/2$-lazy Jerrum-Sinclair chain $P_{\text{zz}}=\frac{1}{2}(P_{\mathrm{JS}}+I)$, where in each step, it performs the transition of $P_{\mathrm{JS}}$ with half probability and stays in the current state with the other half probability. 
It is well-known that $\mu$ is the unique stationary distribution of both $P_{\mathrm{JS}}$ and $P_{\text{zz}}$. 

The mixing time of a Markov chain $P$ with stationary distribution $\mu$ is defined as:
\begin{align*}
T_{\mathrm{mix}}(P,\epsilon) \triangleq \min \left\{t \geq 0 \mid \max_{x \in \Omega} \DTV{P^t(x,\cdot)}{\mu} \leq \epsilon  \right\}\quad \text{ and }\quad T_{\mathrm{mix}}(P)\triangleq T_{\mathrm{mix}}\tp{P,\frac{1}{2\mathrm{e}}},
\end{align*}
where $\DTV{\nu}{\mu} = \frac{1}{2}\sum_{x \in \Omega} \abs{\nu(x) - \mu(x)}$ denotes the total variation distance.

\filbreak
The following result on the rapid mixing of the Jerrum-Sinclair chain is well-known:
\begin{theorem}[Jerrum and Sinclair~\cite{jerrum1989approximating,sinclair1992improved}]\label{thm:Jerrum-Sinclair}
Let $G=(V,E)$ be a simple undirected graph with $n$ vertices and $m$ edges.
Let $\mu$ be the monomer-dimer distribution on $G$ with edge weight $\lambda>0$.
The mixing time of the $1/2$-lazy Jerrum-Sinclair chain $P_{\mathrm{zz}}$ for $\mu$ satisfies
\begin{align*}
    T_{\mathrm{mix}}(P_{\mathrm{zz}}) = {O}_{\lambda}\left(m n^2 \log n\right).
\end{align*}
\end{theorem}
This mixing time bound follows from a lower bound of $\Omega_{\lambda}\left(\frac{1}{mn}\right)$ on the Poincar\'e constant of the Jerrum-Sinclair chain, established via the canonical path method.  The extra $\Theta(n)$ factor in the  Poincar\'e bound arises from its dependence on the length of the canonical paths \cite{sinclair1992improved,diaconis1991geometric}.

Compared to other approaches for analyzing mixing times, the canonical path method offers distinct advantages in many settings where it applies.
For certain fundamental problems, it remains the only known viable approach.
A prominent example is the Jerrum-Sinclair chain for sampling matchings:
A well-known barrier result by Kumar and Ramesh \cite{anil2001coupling} shows that any Markovian coupling for this chain on perfect and near-perfect matchings requires exponential time to coalesce, 
whereas the chain is provably rapidly mixing via the canonical path method~\cite{jerrum1989approximating}.

More recently, a new class of techniques based on high-dimensional expanders (HDX) has led to major advances in analyzing mixing times \cite{anari2019logconcave, alev2020improved, anari2020spectral, chen2021optimal}.
Notably, Chen, Liu and Vigoda \cite{chen2021optimal} established a mixing time bound of ${O}_\lambda(\Delta^{\Delta^{2}} m\log n)$ for the Glauber dynamics of the monomer-dimer model.
For graphs with constant maximum degree $\Delta=O(1)$, this achieves an optimal near-linear mixing time.
However, for general graphs with potentially unbounded $\Delta$, 
the classic bound in \Cref{thm:Jerrum-Sinclair} remains the state-of-the-art.
This reflects an inherent limitation of current HDX-based approaches,  which suffer from an exponential dependence on a correlation decay factor.
For the monomer-dimer model, this decay factor is $\Theta(\sqrt{\lambda\Delta})$ \cite{bayati2007simple}.

Despite the foundational importance and methodological significance  of the Jerrum-Sinclair chain,  little progress has been made in improving its mixing time on general unbounded-degree graphs since \Cref{thm:Jerrum-Sinclair}.
Surpassing this bound appears to require significant innovations in both the canonical path method and newer HDX-based techniques.

In this work, we prove that the Jerrum-Sinclair chain for monomer-dimers mixes substantially faster than the classic bound in \Cref{thm:Jerrum-Sinclair}.
Specifically, we establish the following theorem:

\begin{theorem}\label{thm:main-simple}
Let $G=(V,E)$ be a simple undirected graph with $n$ vertices, $m$ edges, and maximum degree~$\Delta$.
Let $\mu$ be the monomer-dimer distribution on $G$ with edge weight $\lambda>0$.
The mixing time of the $1/2$-lazy Jerrum-Sinclair chain $P_{\mathrm{zz}}$ for $\mu$ satisfies
\begin{align*}
    T_{\mathrm{mix}}(P_{\mathrm{zz}}) = {O}_{\lambda}\left(\Delta m\cdot \min\left\{n,\Delta\log\Delta\cdot \log n \right\} \right)=\widetilde{O}_{\lambda}(\Delta^2 m).
\end{align*}
\end{theorem}
This mixing time bound follows from an improved $\Omega_\lambda(\frac{1}{\Delta m})$ lower bound  on the Poincar\'e constant and a new  ${\Omega}_\lambda(\frac{1}{\Delta^2 m \log\Delta })$ lower bound  on the log-Sobolev constant of the Jerrum-Sinclair chain.  
These bounds are formally restated as \Cref{thm:main} (in \Cref{sec:app-JS-chain}).

By a comparison argument, this also implies a $\widetilde{O}_\lambda(\Delta^3 m)$ mixing time bound for the Glauber dynamics of the monomer-dimer model,  formally stated as \Cref{thm:main-Glauber} (in \Cref{sec:faster-mixing-of-monomer-dimer-glauber-dynamics}).

Our key innovations are twofold, and the techniques are outlined in \Cref{sec:technical-results}:
\begin{itemize}
    \item First, we construct \emph{transport flows} between ``locally coupled'' pairs of states, enabling us to bypass the $\Theta(n)$ factor introduced by the length of canonical paths.
This construction gives rise to certain local variants of the the Poincar\'e and log-Sobolev inequalities.
\item Second, we establish a ``local-to-global'' theorem that lifts these \emph{local functional inequalities} to the  global ones. This local-to-global theorem features an ``{additive}'' accumulation of decay, thus allowing us to bypass the exponential dependence on the decay factor as in HDX.
\end{itemize}

\section{Technical Results and Proof Outline}\label{sec:technical-results}
In this section, we present general results on the mixing times of reversible Markov chains over high-dimensional cubes. 
Let $E$ be a finite ground set and let $[q] = \{0,1,\ldots,q-1\}$ be a finite domain.
Consider a distribution $\mu$ over $[q]^E$, and let $\Omega = \Omega(\mu) \subseteq [q]^E$ denote its support.
Let $Q$ be a Markov chain that is reversible with respect to the stationary distribution $\mu$, meaning that for any $x, y \in \Omega$, the detailed balance equation holds:
$\mu(x)Q(x,y) = \mu(y)Q(y,x)$.
To bound the mixing time of $Q$, we analyze the Poincar\'e and log-Sobolev constants associated with the chain $Q$. 


For any function $f:\Omega \to \mathbb{R}_{\geq 0}$, 
define the variance of $f$ under distribution $\mu$, denoted by $\Var[\mu]{f}$, as the variance of the random variable $F = f(X)$ where $X \sim \mu$, i.e.,
\begin{align}\label{eq:def-variance}
\Var[\mu]{f} \triangleq \Var{F} \triangleq \E{F^2} - \E{F}^2.
\end{align}
Similarly, define the entropy of $f$ under distribution $\mu$, denoted by $\Ent[\mu]{f}$, as:
\begin{align}\label{eq:def-entropy}
\Ent[\mu]{f}\triangleq\Ent{F} \triangleq  \E{F \log F} - \E{F} \log \E{F}.
\end{align}
For any two functions $f,g:\Omega \to \mathbb{R}$, the Dirichlet form associated with the chain $Q$ is given by:
\begin{align*}
  \+E_Q(f,g) \triangleq \frac{1}{2}\sum_{x \in \Omega, y \in \Omega} \mu(x) Q(x,y) \left( f(x) - f(y)\right) \left( g(x) - g(y)\right).
\end{align*}
The Poincar\'e constant $\gamma(Q)$ of $Q$ is defined as:
\begin{align*}
\gamma(Q) = \inf \left\{\frac{\+E_Q(f,f)}{\Var[\mu]{f}} \mid f :\Omega \to \mathbb{R} \land \Var[\mu]{f} > 0\right\}.
\end{align*}
The log-Sobolev constant $\rho(Q)$ of $Q$ is defined as:
\begin{align*}
\rho(Q) = \inf \left\{\frac{\+E_Q(f,f)}{\Ent[\mu]{f^2}} \mid f :\Omega \to \mathbb{R} \land \Ent[\mu]{f^2} > 0\right\}.
\end{align*}
The mixing time of $Q$ can be bounded using the Poincar\'e constant as follows:
\begin{align*}
T_{\mathrm{mix}}\tp{\frac{Q+I}{2},\epsilon} = O\tp{\frac{1}{\gamma(Q)}\tp{\log\frac{1}{\mu_{\min}} + \log\frac{1}{\epsilon}}},
\end{align*}
where $\mu_{\min} = \min_{x \in \Omega} \mu(x)$ and $\frac{Q+I}{2}$ corresponds to the $\frac{1}{2}$-lazied version of the chain $Q$. 
If $Q$ is positive semidefinite, this mixing time bound also holds for $Q$.

The log-Sobolev constant $\rho(Q)$ can provide a sharper bound on the mixing time, specifically: 
\begin{align*}
  T_{\mathrm{mix}}\tp{\frac{Q+I}{2}, \epsilon} = O\tp{\frac{1}{\rho(Q)}\tp{\log\log \frac{1}{\mu_{\min}} + \log\frac{1}{\epsilon}}}.
\end{align*}
This bound on the mixing time, derived from the log-Sobolev constant, also holds for $Q$.

\subsection{Rapid mixing via local functional inequalities}
Our results establish both the Poincar\'e and log-Sobolev inequalities by first introducing new local versions of these inequalities. We then develop a ``local-to-global'' argument to lift these functional inequalities from the local to the global setting.

Next, we introduce the new \emph{local} versions of functional inequalities for mixing times.
\begin{definition}[local functional inequalities] \label{def:HS-route}
  Let $\mu$ be a distribution supported over $\Omega \subseteq [q]^E$. 
  Let $Q$ be a reversible chain on $\Omega$ with stationary distribution $\mu$.
  \begin{itemize}
  \item (local Poincar\'e inequality) We say that ${Q}$ satisfies the $\alpha$-\emph{local Poincar\'e inequality} if 
  \begin{align} \label{eq:def-HS-route-V}
    \alpha \cdot \sum_{e\in E} \Var{\E{\*F \mid \*X_e}} \leq  \+E_{Q}(f,f),\quad \forall f: \Omega(\mu) \to \mathbb{R},
  \end{align}
  where $\*X \sim \mu$ and $\*F \triangleq  f(\*X)$.

  \item   (local log-Sobolev inequality) We say $Q$ satisfies the \emph{$\alpha$-local log-Sobolev inequality} if 
  \begin{align} \label{eq:def-HS-route-E}
    \alpha \cdot \sum_{e\in E} \Ent{\E{\*F^2 \mid \*X_e}} \leq  \+E_{Q}(f,f),\quad \forall f: \Omega(\mu) \to \mathbb{R}.
  \end{align}
  where $\*X \sim \mu$ and $\*F \triangleq  f(\*X)$.
  \end{itemize}
\end{definition}
Note that in~\eqref{eq:def-HS-route-V} and~\eqref{eq:def-HS-route-E}, $\E{\*F \mid \*X_e}$ and $\E{\*F^2 \mid \*X_e}$ are both real-valued random variables. The operators $\Var{\cdot}$ and $\Ent{\cdot}$ are as defined in~\eqref{eq:def-variance} and~\eqref{eq:def-entropy}, respectively.




Similar to other mixing properties, such as strong spatial mixing or spectral independence, local functional inequalities are particularly useful when they hold under arbitrary pinnings.
Let $\Lambda \subseteq E$ be a fixed subset. 
A pinning $\tau \in [q]^{E \setminus \Lambda}$ is called feasible if $\tau \in \Omega(\mu_{E \setminus \Lambda})$, where $\mu_{E \setminus \Lambda}$ is the marginal distribution on $E \setminus \Lambda$ induced from $\mu$, and $\Omega(\mu_{E \setminus \Lambda})$ represents its support. 
Next, define the conditional distribution $\mu^\tau$, which is the distribution of $X \sim \mu$ conditioned on $X_{E \setminus \Lambda} = \tau$.
This distribution is over $[q]^E$, where the values of the variables in $E \setminus \Lambda$ are fixed to $\tau$. 

We consider a family of reversible Markov chains, each corresponding to a conditional distribution $\mu^\tau$ induced by $\mu$ given a feasible pinning $\tau$. Specifically, suppose we have a family of chains:
\begin{align*}
  \mathfrak{Q}=\mathfrak{Q}(\mu) &= \left\{Q^\tau \mid \tau \text{ is a feasible pinning}\right\},
\end{align*}
where each $Q^\tau$ corresponds to a reversible Markov chain with stationary distribution $\mu^\tau$, 
and the family $\mathfrak{Q}$ encompasses all subsets $\Lambda \subseteq E$ and all feasible pinnings $\tau \in [q]^{E \setminus \Lambda}$.

The following defines local functional inequalities under arbitrary pinning.
\begin{definition}[local functional inequalities under pinnings] 
\label{def:HS-route-family}
Let $\mathfrak{Q}$ be a family of reversible chains corresponding to the conditional distributions induced by a distribution $\mu$ supported on $\Omega \subseteq [q]^E$.

  We say that $\mathfrak{Q}$ satisfies $(\alpha_1,\alpha_2,\ldots,\alpha_{\abs{E}})$-\emph{local Poincar\'e inequalities} 
  if for any non-empty $\Lambda \subseteq E$ and any feasible pinning $\tau \in [q]^{E \setminus \Lambda}$, 
  the chain $Q^\tau\in\mathfrak{Q}$ satisfies the $\alpha_{|\Lambda|}$-local Poincar\'e inequality.

  Similarly, $\mathfrak{Q}$ satisfies $(\alpha_1,\alpha_2,\ldots,\alpha_{\abs{E}})$-\emph{local log-Sobolev inequalities} 
  if for any non-empty $\Lambda \subseteq E$ and any feasible pinning $\tau \in [q]^{E \setminus \Lambda}$,
  the chain $Q^\tau\in\mathfrak{Q}$ satisfies the $\alpha_{|\Lambda|}$-local log-Sobolev inequality.

\end{definition}


Intuitively, we require all chains in the family $\mathfrak{Q}$ are, in some sense, ``the same'' Markov chain. 
This is captured by the following significantly relaxed notion of concavity for Dirichlet forms.

\begin{definition}[concave Dirichlet forms]\label{def:concave-Dirichlet-forms}
  Let $\mu$ be a distribution over $[q]^E$ with support $\Omega$.
  Let $\mathfrak{Q}$ be a family of reversible chains corresponding to the conditional distributions induced by $\mu$.
  We say that $\mathfrak{Q}$ has \emph{concave Dirichlet forms} if, 
  for any $Q^\tau \in \mathfrak{Q}$, where $\tau \in [q]^{E \setminus\Lambda}$, and any $f:\Omega(\mu^\tau) \to \mathbb{R}$,
    \begin{align*}
   \frac{1}{|\Lambda|} \sum_{e \in \Lambda} \E[c \sim \mu^\tau_e]{\+E_{Q^{\tau \land (e \gets c)}}(f,f)} \leq \+E_{Q^\tau}(f,f),
  \end{align*}
  where $\tau \land (e \gets c)$ denotes the pinning obtained from extending $\tau$ by fixing the value of $e$ to $c$.
\end{definition}

Our main technical result is the following ``local-to-global'' theorem for functional inequalities.

\begin{theorem}\label{thm:general-results}
  Let $\mu$ be a distribution over $[q]^E$ with support $\Omega$, where $m=|E|$.
  Let $\mathfrak{Q}$ be a family of reversible Markov chains corresponding to the conditional distributions induced by $\mu$, and let $Q \in \mathfrak{Q}$ be the chain with stationary distribution $\mu$.
  Suppose that $\mathfrak{Q}$ has concave Dirichlet forms.
  \begin{itemize}
  \item If $\mathfrak{Q}$ satisfies $(\alpha_1,\alpha_2,\ldots,\alpha_{m})$-local Poincar\'e inequalities with $\alpha_i > 0$ for all $i$, then the Poincar\'e constant $\gamma(Q)$ of $Q$ satisfies 
  $$\gamma(Q)\ge \tp{\sum_{k=1}^{m} \frac{1}{k \alpha_{k}}}^{-1}.$$
  \item If $\mathfrak{Q}$ satisfies $(\alpha_1,\alpha_2,\ldots,\alpha_{m})$-local log-Sobolev inequalities with $\alpha_i > 0$ for all $i$, then the log-Sobolev constant $\rho(Q)$ of $Q$ satisfies 
  $$\rho(Q)\ge \tp{\sum_{k=1}^{m} \frac{1}{k  \alpha_{k}}}^{-1}.$$
  \end{itemize}
\end{theorem}


\begin{remark}\label{rem:general-results-remark}
Compared to other existing local-to-global theorems established for high-dimensional expanders \cite{anari2019logconcave, cryan2021modified,alev2020improved,anari2020spectral}, 
which typically exhibit a multiplicative accumulation of decay, \Cref{thm:general-results} achieves an additive accumulation of decay.
In particular, if $\mathfrak{Q}$ satisfies an $(\alpha_1,\alpha_2,\ldots,\alpha_{m})$-local functional inequality with $\alpha_i = \Omega\tp{\frac{1}{a\cdot i}}$ for some $a>0$, then by \Cref{thm:general-results}, the corresponding Poincar\'e or log-Sobolev constant of the chain $Q$ is at least $\Omega(\frac{1}{a\cdot m})$.
\end{remark}

\Cref{thm:general-results}  is proved in \Cref{sec:local-to-global}.
\Cref{thm:general-results} applies to a broad class of Markov chains. 
For instance, it can be applied to the family of Metropolis-Hastings dynamics,  
where each $Q^\tau$ corresponds to the Metropolis-Hastings dynamics for the conditional distribution $\mu^\tau$, 
as well as for the family of Glauber dynamics, 
where each $Q^\tau$ corresponds to the Glauber dynamics for $\mu^\tau$. 
For these families of chains, the concavity of the Dirichlet forms can be verified straightforwardly.

   In \Cref{sec:localization-scheme}, we further present a generalization of \Cref{thm:general-results} to localization schemes.


\subsection{Local functional inequalities via transport flow}

We develop a constructive approach for establishing the local functional inequalities in \Cref{def:HS-route} using a technique we call \emph{transport flow}, which extends the canonical path method.

Given a sequence of states $\gamma = (x_0,x_1,\ldots,x_{\ell})$  from the state space $\Omega$,
we denote its starting point as $s(\gamma) = x_0$ and its endpoint as $t(\gamma) = x_\ell$.
We call $\gamma$ a path in the transition graph of $Q$ if $Q(x_i,x_{i+1}) > 0$ for all $i < \ell$.
We denote the length (number of edges) of this path $\gamma$ by $\ell(\gamma)=\ell$.
For two states $x, y \in \Omega(\mu)$, if $Q(x,y) > 0$, we say there is a transition from $x$ to $y$, denoted as $(x \mapsto y)$.
Given a transition $(x \mapsto y)$, we write $(x \mapsto y) \in \gamma$  to indicate that there exists some $i < \ell$ such that $x = x_i$ and $y = x_{i+1}$.

\begin{definition}[transport flow]\label{def:transport-flow}
    Let $\nu$ and $\pi$ be two distributions over the state space $\Omega$, 
    and let $Q$ be a Markov chain on $\Omega$. 
    The \emph{transport flow} from $\nu$ to $\pi$ is a distribution $\Gamma$ over paths on the transition graph of $Q$ that satisfies:
    \begin{itemize}
    \item The starting point $s(\gamma)$ of a path $\gamma \sim \Gamma$ follows the distribution $\nu$;
    \item The endpoint $t(\gamma)$ of a path $\gamma \sim \Gamma$ follows the distribution $\pi$.
    \end{itemize}
\end{definition}

The random paths in $\Gamma$ transport the distribution $\nu$ to $\pi$ using the transitions of the Markov chain $Q$.
Given a random path $\gamma \sim \Gamma$, its starting point and endpoint $(s(\gamma),t(\gamma))$ naturally form a coupling of the two distributions $\nu$ and $\pi$.

\begin{remark}[Canonical path and multicommodity flow as transport flow]\label{rem:canonical-path-as-transport-flow}
  The classical canonical path method, and more generally, multicommodity flow~\cite{diaconis1991geometric,sinclair1992improved}, can be interpreted as constructing a transport flow from $\mu$ to $\mu$ itself. 
  Specifically, these methods define a collection of paths $\+P$, where each path $\gamma \in \+P$ is assigned with a weight $w(\gamma)$,
  such that for any $x, y \in \Omega$, the total weight $w(\gamma)$ for all paths from $x$ to $y$ is precisely $\mu(x)\mu(y)$. 
  This corresponds to a transport flow $\Gamma$ from $\mu$ to $\mu$, where each path $\gamma \in \+P$ is sampled with probability $w(\gamma)$. 
  Thus, a random path $\gamma \sim \Gamma$ induces an \emph{independent coupling} with itself through its endpoints $(s(\gamma), t(\gamma))$.
\end{remark}


The following theorem establishes local functional inequalities based on the existence of a family of transport flows with bounded average congestion and average squared length.

For any $e \in E$ and $c \in [q]$, let $\mu^{e \gets c}$ denote the distribution of $\*X \sim \mu$ conditioned on $\*X_e = c$.

\begin{theorem}[\HSRV{} via transport flow] \label{thm:HS-route-via-congestion}
  Let $\kappa,L > 0$. 
  Let $\mu$ be a distribution over $[q]^E$ with support $\Omega$.
  Let $Q$ be a reversible Markov chain on $\Omega$ with stationary distribution $\mu$.
  Suppose there exists a family of transport flows $\{\Gamma_e^{a\to b}\}$ from $\mu^{e \gets a}$ to $\mu^{e \gets b}$ for all $e \in E$ and $a,b \in \Omega(\mu_e)$ with $a < b$, satisfying the following conditions:
  \begin{itemize}
  \item ($\kappa$-expected congestion) For any transition $(x \mapsto y)$, and any $a,b \in [q]$ with $a < b$,
    \begin{align} \label{eq:def-congestion}
    \sum_{e\in E:a,b\in \Omega(\mu_e)} \mu_e(a)\mu_e(b) \E[\gamma \sim \Gamma_e^{a\to b}]{\frac{\*1[(x \mapsto y) \in \gamma]}{\ell(\gamma)}} \leq \kappa \cdot \mu(x)Q(x,y).
    \end{align}
  \item ($L$-expected squared length)  For any $e \in E$, and any $a,b \in \Omega(\mu_e)$ with $a < b$,
  \begin{align*}
  \E[\gamma \sim \Gamma_e^{a\to b}]{\ell(\gamma)^2} \leq L.
  \end{align*}
  \end{itemize}
  Then, the Markov chain $Q$ satisfies the $\frac{1}{2q^2\kappa L}$-local Poincar\'e inequality.
\end{theorem}

\begin{theorem}[local log-Sobolev inequality via transport flow] \label{thm:log-Sob-HS-route-via-congestion}
  Let $\kappa > 0$. 
  Let $\mu$ be a distribution over $[q]^E$ with support $\Omega$.
  Let $Q$ be a reversible Markov chain on $\Omega$ with stationary distribution $\mu$.
  Suppose there exists a family of transport flows $\{\Gamma_e^{a\to b}\}$ from $\mu^{e \gets a}$ to $\mu^{e \gets b}$ for all $e \in E$ and $a,b \in \Omega(\mu_e)$ with $a < b$, satisfying the following condition:
  \begin{itemize}
  \item (strong $\kappa$-expected congestion) For any transition $(x \mapsto y)$, and any $a,b \in [q]$ with $a < b$,
    \begin{align} \label{eq:def-strong-congestion}
     \sum_{e\in E: a,b \in \Omega(\mu_e) } \mu_e(a) \mu_e(b) \E[\gamma \sim \Gamma^{a\to b}_e]{\ell(\gamma) \*1[(x \mapsto y) \in \gamma]} \leq \kappa \cdot \mu(x)Q(x,y).
    \end{align}
  \end{itemize}
  Then, the Markov chain $Q$ satisfies the $\alpha$-local log-Sobolev inequality with
  \begin{align*}
  \alpha = \frac{1}{2 q^2 \kappa} \cdot \frac{{1-2 \phi}}{\log(\frac{1}{\phi} - 1)},
  \end{align*}
  where $\phi \le \min \set{\mu_e(c) \mid e \in E, c\in [q], \mu_e(c) \neq 0}$ is the marginal lower bound.
\end{theorem}

\Cref{thm:HS-route-via-congestion,thm:log-Sob-HS-route-via-congestion} are proved in \Cref{sec:transport-flow}.
Compared to the classical canonical path and multicommodity flow approaches, 
these theorems establish \emph{local} functional inequalities for mixing times based on the \emph{average} congestion and length of paths between \emph{coupled} pairs. 
In our applications, we carefully exploit this coupling induced by transport flow  to ensure that the congestion and length remain well bounded.

A similar coupling-based enhancement of the canonical path argument was recently used in \cite{CCFV25} to analyze the Glauber dynamics for edge colorings on trees. 
Our definition of transport flow is more general.
More importantly, we show that transport flow can be utilized to establish local functional inequalities that guarantee rapid mixing in general  Markov chains.

\subsection{Application to the Jerrum-Sinclair chain}\label{sec:app-JS-chain}
We then apply our general approach to the Jerrum-Sinclair chain for the monomer-dimer model and establish the following result, which is a formal restatement of \Cref{thm:main-simple}.
\begin{theorem}\label{thm:main}
Let $G=(V,E)$ be a simple undirected graph with $n$ vertices, $m$ edges, and maximum degree $\Delta$.
Let $\mu$ be the monomer-dimer distribution on $G$ with edge weight $\lambda>0$, and define $\overline{\lambda} = \max\set{1,\lambda}$.
The Jerrum-Sinclair chain $P_{\mathrm{JS}}$ for $\mu$ has Poincar\'e constant $ \gamma(P_{\mathrm{JS}})$ and log-Sobolev constant $\rho(P_{\mathrm{JS}})$ satisfying:
\begin{align*}
    \gamma(P_{\mathrm{JS}}) = \Omega\tp{\frac{1}{\overline{\lambda}^{3}\log^2(1+\ol{\lambda}) \cdot  \Delta \cdot m}} \quad \text{and} \quad \rho(P_{\mathrm{JS}}) = \Omega\tp{\frac{1}{\overline{\lambda}^{4} \cdot \Delta^2 \log(\ol{\lambda} \Delta) \cdot m}}.
\end{align*}
As a consequence, the mixing time of the $1/2$-lazy Jerrum-Sinclair chain $P_{\mathrm{zz}}=\frac{1}{2}(P_{\mathrm{JS}}+I)$ satisfies
\begin{align*}
    T_{\mathrm{mix}}(P_{\mathrm{zz}}) &= O\left(\overline{\lambda}^{3} \Delta m\cdot  \min \left\{n\log^2(1+\ol{\lambda}),\,\, \overline{\lambda} \Delta \log(\ol{\lambda} \Delta)\cdot \log n \right\}\right).
\end{align*}
\end{theorem}
In comparison, the original bound by Jerrum and Sinclair \cite{jerrum1989approximating,sinclair1992improved,jerrum2003counting} was $O(\ol{\lambda}mn^2\log n)$, which remained the best known bound for general graphs with arbitrary maximum degree.

\paragraph{The Jerrum-Sinclair chain on monomer-dimers.}
We can   interpret the monomer-dimer distribution $\mu$ on graph $G=(V,E)$, defined over $2^E$, equivalently as a distribution over $\{0,1\}^E$, where each subset $M \subseteq E$  corresponds to a configuration $\sigma_M \in \{0,1\}^E$, with $\sigma_M(e)$ indicating whether the edge $e$ is occupied in $M$.

Let $\mathfrak{Q}_\JS=\mathfrak{Q}_\JS(\mu)$ denote the family of Jerrum-Sinclair chains for the monomer-dimer distributions $\mu^{\tau}$ induced by $\mu$.
%
For each feasible pinning $\tau \in \{0,1\}^{E \setminus \Lambda}$, where $\Lambda \subseteq E$ is a subset of edges, the Markov chain $Q^\tau\in \mathfrak{Q}_\JS$ is a non-lazy Jerrum-Sinclair chain for the conditional distribution $\mu^\tau$. 

Specifically, given the current state $X_t \in \{0,1\}^{E}$, the chain $Q^\tau$ transitions to $X_{t+1}$ as follows:
\begin{enumerate}
\item Select an edge $e = \{u,v\} \in \Lambda$ uniformly at random;
\item Propose a candidate matching $M$ for $X_{t+1}$ 
using the \emph{down}, \emph{up}, and \emph{exchange} transitions as described in \Cref{sec:introduction};
\item Accept $M$ as $X_{t+1}$ with probability $\min\set{1,\frac{\mu^\tau(M)}{\mu^\tau(X_t)}}$; otherwise, set $X_{t+1}=X_{t}$.
\end{enumerate}


The Jerrum-Sinclair chains satisfy the concavity condition stated in \Cref{def:concave-Dirichlet-forms}.
\begin{proposition}\label{lem:concave-JS}
The family $\mathfrak{Q}_{\JS}$ of Jerrum-Sinclair chains has concave Dirichlet forms.
\end{proposition}

The proof of \Cref{lem:concave-JS} follows from a straightforward calculation, provided in \Cref{sec:concave-Dirichlet-forms}.

Our main objective is to establish the local functional inequalities in \Cref{def:HS-route-family} for $\mathfrak{Q}$. 
\begin{lemma}\label{lem:local-functional-inequalities}
The family $\mathfrak{Q}_\JS$ of Jerrum-Sinclair chains satisfies the following local functional inequalities:
\begin{enumerate}
\item $(\alpha_1,\alpha_2,\ldots,\alpha_{m})$-local Poincar\'e inequalities with $\alpha_k = \Omega\tp{\frac{1}{\ol{\lambda}^3\log^2(1+\ol{\lambda})\cdot \Delta  \cdot k}}$;
\item $(\alpha_1,\alpha_2,\ldots,\alpha_m)$-local log-Sobolev inequalities with $\alpha_k = \Omega\tp{\frac{1}{\ol{\lambda}^4\cdot \Delta^2 \log(\ol{\lambda} \Delta) \cdot k}}$.
\end{enumerate}
\end{lemma}

\Cref{thm:main} then follows from \Cref{lem:concave-JS}, \Cref{lem:local-functional-inequalities}, and \Cref{thm:general-results}.

\paragraph{The construction of transport flow.}
\Cref{lem:local-functional-inequalities} is proved by constructing a family of transport flows (\Cref{def:transport-flow}) and applying \Cref{thm:HS-route-via-congestion,thm:log-Sob-HS-route-via-congestion}. 
Specifically, for any feasible pinning $\tau \in \{0,1\}^{E \setminus \Lambda}$ and
any edge $e \in \Lambda$ with $|\Omega(\mu^\tau_e)| = 2$, we construct a transport flow from $\mu^{\tau \wedge (e \gets 0)}$ to $\mu^{\tau \wedge (e \gets 1)}$ 
that satisfies the congestion and path-length conditions required by \Cref{thm:HS-route-via-congestion,thm:log-Sob-HS-route-via-congestion}.

For convenience, we present the construction under no pinning, i.e., while $\tau=\emptyset$. The general case follows by self-reducibility.
The transport flow from $\mu^{e \gets 0}$ to $\mu^{e \gets 1}$  is constructed as follows:
\begin{itemize}
\item \textbf{Coupling step}: Sample a pair of configurations $(X,Y)$ from a coupling of $\mu^{e \gets 0}$ and $\mu^{e \gets 1}$;
\item \textbf{Canonical path step}: Construct a deterministic canonical path $\gamma$ from $X$ to $Y$ using the transitions of $Q$, following the construction in~\cite{jerrum1989approximating}.
\end{itemize}
The distribution $\Gamma_e=\Gamma_{e}^{0\to 1}$ of the resulting path $\gamma$ defines the transport flow from $\mu^{e \gets 0}$ to $\mu^{e \gets 1}$. 
Unlike the classical construction of canonical paths~\cite{jerrum1989approximating},
our approach constructs a path only between the configurations $(X,Y)$ generated by the coupling, rather than for all possible pairs.

To achieve this, we introduce the \emph{local flipping coupling} process for monomer-dimers.
\begin{construction}[local flipping coupling]\label{def:new-coupling-2}
  Let $G=(V,E)$ and $e \in E$.
  Let $\mu$ be  the monomer-dimer  distribution on $G$.  
  The local flipping coupling $\+C_e$ of $\mu^{e \gets 0}$ and $\mu^{e \gets 1}$ is constructed as follows: 
  \begin{enumerate} 
  \item Sample $Y \sim \mu^{e \gets 1}$ and $Z \sim \mu^{e \gets 0}$ independently.
  \item Decompose the differences between $Y$ and $Z$ into paths and cycles. 
  Let $B$ be the connected component (either a path or a cycle of even length) that contains the edge $e$.
  \item Construct $X = Z_B \uplus Y_{E \setminus B}$, the configuration obtained by concatenating $Z_B$ and $Y_{E \setminus B}$. 
  The outcome of the coupling is the pair $(X,Y)$.
  \end{enumerate}
\end{construction}

The matching $X$ is obtained from $Y$ by flipping a single alternating path or cycle $B$, which corresponds to the disagreeing component between $Y$ and $Z$ that contains the edge $e$.    
We will verify (in \Cref{lem:explicit-formula}) that the procedure in \Cref{def:new-coupling-2} indeed defines a valid coupling between $\mu^{e \gets 0}$ and $\mu^{e \gets 1}$. 
The pair $(X,Y)$ generated by this coupling differs only on $B$. 

To complete the construction of the transport flow, 
we construct a canonical path from $X$ to $Y$, following the classic strategy in~\cite{jerrum1989approximating}. 
We also assume a total ordering of the vertices in $G$.



\begin{construction}[transport flow for monomer-dimer distributions] \label{def:monomer-dimer-trans-path}
Let $G=(V,E)$ and $e \in E$.
The transport flow $\Gamma_e$ from $\mu^{e \gets 0}$ to $\mu^{e \gets 1}$ is constructed as follows: 
Sample $(X,Y)\sim\+C_e$ according to the local flipping coupling $\+C_e$ defined in \Cref{def:new-coupling-2}, 
and let $B =X\oplus Y\triangleq  \{f \in E \mid X_f \neq Y_f\}$ be the set of edges where $X$ and $Y$ disagree. 
\begin{enumerate}
\item If $B$ forms a path, let $w^*$ be the larger endpoint of the path $B$, and let $e_1 \in B$ be the unique edge incident to $w^*$. Number the remaining edges along the path as $e_2,e_3,\ldots,e_{|B|}$. 
\begin{enumerate}
\item If $X_{e_1} = 1$ (i.e., $e_1$ is in the matching $X$), 
first apply a down transition to remove $e_1$, 
then perform exchange transitions on each pair $(e_{2i},e_{2i+1})$ along the path for $1\leq i \leq \lfloor (|B|-1)/2 \rfloor$, 
and finally, apply an up transition on $e_{|B|}$ if $|B|$ is even. 
\item If $X_{e_1} = 0$ (i.e., $e_1$ is not in $X$), apply exchange transitions on each pair $(e_{2i-1},e_{2i})$ along the path for $1\leq i \leq \lfloor |B|/2 \rfloor$, and conclude with an up transition on $e_{|B|}$ if $|B|$ is odd.
\end{enumerate}

\item If $B$ forms an even-length cycle, 
let $w^*$ be the largest vertex in $B$, 
and let $e_1 = \{w,w^*\} \in X$ be the unique edge incident to $w^*$ in the matching $X$. 
Number the edges along the cycle in the direction $w^* \to w$ as $e_2,e_3,\ldots,e_{|B|}$, where $e_{|B|}$ is incident to $w^*$. 
Apply a down transition to remove $e_1$, 
perform exchange transitions on each pair $(e_{2i},e_{2i+1})$ along the cycle for $1\leq i \leq \frac{|B|}{2} - 1$,
and conclude with an up transition on $e_{|B|}$.
\end{enumerate}
\end{construction}

The following theorem shows the effectiveness of the transport flow in \Cref{def:monomer-dimer-trans-path}.
\begin{theorem} \label{thm:effective-trans-flow}
    Let $G=(V,E)$ be a simple undirected graph with $m=|E|$ edges and maximum degree $\Delta$, 
    and let $\mu$ be the monomer-dimer distribution on graph $G$ with edge weight $\lambda > 0$. Define $\ol{\lambda}=\max\{1, \lambda\}$.
    For every edge $e\in E$, the transport flow $\Gamma_e$ in~\Cref{def:monomer-dimer-trans-path} satisfies the following conditions:
    \begin{enumerate}
    \item\label{thm:trans-flow-length} $O(\ol{\lambda}\Delta  \log^2(1+\ol{\lambda}))$-expected squared length;
    \item\label{thm:trans-flow-congestion} $O(\ol{\lambda}^2  m)$-expected congestion;
    \item\label{thm:trans-flow-strong-congestion} strong $O(\ol{\lambda}^4 \Delta^2 m)$-expected congestion;
    \end{enumerate}
    as required in \Cref{thm:HS-route-via-congestion,thm:log-Sob-HS-route-via-congestion}.
\end{theorem}

Then, \Cref{lem:local-functional-inequalities} follows almost immediately from a combination of \Cref{thm:HS-route-via-congestion,thm:log-Sob-HS-route-via-congestion,thm:effective-trans-flow}, utilizing the  marginal lower bound $\phi=\Omega\left(\frac{\lambda}{(1+\lambda\Delta)^2}\right)$ known for monomer-dimer distributions. 
The only remaining technical step is to extend the bounds in \Cref{thm:effective-trans-flow} to transport flows under arbitrary pinning. This, along with a formal proof of \Cref{thm:effective-trans-flow}, is the main focus of \Cref{sec:app-Jerrume-Sinclair}.

\section{General Approach to Mixing Times}
In this section, we introduce a systematic framework for analyzing the mixing times of general reversible Markov chains on $\Omega\subseteq[q]^E$.
Our approach consists of two key components:
In \Cref{sec:local-to-global}, 
    we establish a ``local-to-global''  argument that lifts local Poincar\'e and log-Sobolev inequalities to their global counterparts that imply mixing time bounds.
In \Cref{sec:transport-flow}, we show how these local functional inequalities  can be derived through the construction of transport flows.

\subsection{Rapid mixing via local functional inequalities}\label{sec:local-to-global}


We prove \Cref{thm:general-results}, a ``local-to-global'' theorem for the Poincar\'e and log-Sobolev inequalities. 

A key step in the proof is a one-step comparison lemma.
Let $\mu$ be a distribution supported on $\Omega \subseteq [q]^E$. 
Suppose $\mathfrak{Q}$ is a family of reversible Markov chains for the conditional distributions induced by $\mu$. 
Consider the chain $Q^\tau \in\mathfrak{Q}$ for $\mu^\tau$ with feasible pinning $\tau \in [q]^{E \setminus \Lambda}$.
The following lemma compares the chain $Q^\tau$ to the chains $Q^{\tau \land (e \gets c)}$ for all $\mu^{\tau \land (e \gets c)}$ with  $e \in  \Lambda$ and  $c \in \Omega(\mu_e^\tau)$.

\begin{lemma}[one-step comparison lemma]\label{lem:one-step-comparison}
Suppose the Dirichlet forms satisfy
\begin{align}\label{eq:lemma-Dirichlet}
 \frac{1}{|\Lambda|} \sum_{e \in \Lambda} \E[c \sim \mu_e^\tau ]{\+E_{Q^{\tau \land (e \gets c)}}(f,g)} \leq \+E_{Q^\tau}(f,g), \quad \forall f,g: \Omega \to \mathbb{R}.
\end{align}
Then, the following one-step comparison results hold:
\begin{itemize}
  \item (Poincar\'e inequality) Suppose the Poincaré constants satisfy   $\gamma(Q^{\tau \land (e \gets c)})\ge\gamma_0$ for all chains $Q^{\tau \land (e \gets c)}$, 
  and the chain $Q^\tau$ satisfies the $\alpha$-local Poincar\'e inequality.
  Then, the Poincar\'e constant of $Q^\tau$ satisfies 
  \[\frac{1}{\gamma(Q^\tau)}\leq \frac{1}{\gamma_0} + \frac{1}{|\Lambda|\alpha}.\]
  \item (log-Sobolev inequality) Suppose the log-Sobolev constants satisfy $\rho(Q^{\tau \land (e \gets c)})\ge\rho_0$ for all $Q^{\tau \land (e \gets c)}$, 
  and $Q^\tau$ satisfies the $\alpha$-local log-Sobolev inequality.
  Then, the log-Sobolev constant of $Q^\tau$ satisfies
 \[\frac{1}{\rho(Q^\tau)}\leq \frac{1}{\rho_0} + \frac{1}{|\Lambda|\alpha}.\]
\end{itemize}
\end{lemma}
\begin{proof}
  Let $\*X \sim \mu^\tau$ be a random configuration, and define $\*F \triangleq  f(\*X)$, where $f:\Omega \to \mathbb{R}$ is an arbitrary function.
  Since $\tau$ fixes the values on $E \setminus \Lambda$, 
  we apply the law of total variance for each free variable $e \in \Lambda$ and take the average, yielding
\begin{align}\label{eq:law-tot-var}
    \Var{\*F} 
    &= \frac{1}{|\Lambda|} \sum_{e\in \Lambda} \Var{\E{\*F \mid \*X_e}} + \frac{1}{|\Lambda|} \sum_{e \in \Lambda} \E{\Var{\*F \mid \*X_e}}.
\end{align}
Thus, to bound $\Var[\mu^\tau]{f} = \Var{\*F}$, it suffices to bound each term on the right-hand side separately.
Since $Q^\tau$ satisfies the $\alpha$-local Poincar\'e inequality, we have
\begin{align} \label{eq:IH-first-term}
    \sum_{e\in  \Lambda} \Var{\E{\*F \mid \*X_e}} \leq  \frac{1}{\alpha} \cdot \+E_{Q^\tau}(f, f).
\end{align}
For the second term, by the assumption on the Poincar\'e constants for all $Q^{\tau \land (e \gets c)}$ and~\eqref{eq:lemma-Dirichlet}, we get
\begin{align} \label{eq:IH-second-term}
    \sum_{e \in \Lambda}\E{\Var{\*F \mid \*X_e}} &= \sum_{e \in \Lambda} \sum_{c \in \Omega(\mu_e)} \mu^\tau_e(c) \cdot \Var[\*X \sim \mu^{\tau \land (e \gets c)}]{f(\*X)}\notag\\
    &\leq  \frac{1}{\gamma_0} \sum_{e \in \Lambda} \E[c \sim \mu^\tau_e]{\+E_{Q^{\tau \land (e \gets c)}}(f,f)}\notag\\
    &\leq \frac{|\Lambda|}{\gamma_0}\cdot \+E_{Q^\tau}(f, f).
\end{align}
Combining \eqref{eq:law-tot-var}, \eqref{eq:IH-first-term}, and \eqref{eq:IH-second-term}, we have
\begin{align*}
    \Var{\*F}
    &\leq \frac{1}{|\Lambda|} \cdot \frac{1}{\alpha} \cdot \+E_{Q^\tau}(f, f) + \frac{1}{\gamma_0} \cdot \+E_{Q^\tau}(f, f)
    = \tp{\frac{1}{\gamma_0} + \frac{1}{|\Lambda| \alpha}}\cdot \+E_{Q^\tau}(f, f).
\end{align*}
This establishes the bound on the Poincar\'e constant.

For the log-Sobolev constant, we apply the law of total entropy:
\begin{align*}
  \Ent{\*F^2} 
  &= \frac{1}{\abs{\Lambda}} \sum_{e\in \Lambda} \Ent{\E{\*F^2 \mid \*X_e}} + \frac{1}{\abs{\Lambda}} \sum_{e \in \Lambda} \E{\Ent{\*F^2 \mid \*X_e}}.
\end{align*}
The remaining steps follow analogously to the proof of Poincar\'e constant.
\end{proof}

\begin{proof}[Proof of~\Cref{thm:general-results}]
We begin with the first part of \Cref{thm:general-results}, concerning the Poincar\'e inequality.
Let $m = |E|$.
For any feasible pinning $\tau\in[q]^{E \setminus \Lambda}$, where $\Lambda \subseteq E$ is arbitrary, we aim to establish:
\begin{align} \label{eq:IH-half-ineq}
    \Var[\mu^\tau]{f} &\leq \sum_{i=1}^{|\Lambda|}\frac{1}{i\alpha_{i}} \cdot \+E_{Q^\tau}(f, f),
\end{align}
where $\alpha_i$ for $1 \leq i \leq m$ are the parameters of local Poincar\'e inequalities assumed in \Cref{thm:general-results}.

We proceed by induction on $|\Lambda|$.
For the base case $\Lambda = \emptyset$, $\mu^\tau$ is a Dirac measure, implying $\Var[\mu^\tau]{f} = 0$, so \eqref{eq:IH-half-ineq} holds trivially.

Now assume \eqref{eq:IH-half-ineq} holds for every $\Lambda' \subseteq E$ with $\abs{\Lambda'} < k$ and every $\tau' \in \Omega(\mu_{ E \setminus \Lambda'})$. 
We show that it also holds for every $\Lambda$ with $|\Lambda|=k$ and every $\tau \in \Omega(\mu_{E \setminus \Lambda})$. 
Applying the induction hypothesis, \Cref{def:HS-route-family}, \Cref{def:concave-Dirichlet-forms,lem:one-step-comparison} (with parameter $\alpha=\alpha_{k}$), 
we obtain
\begin{align*}
  \frac{1}{\gamma(Q^\tau)} \leq \sum_{i=1}^{k-1}\frac{1}{i\alpha_{i}} + \frac{1}{\alpha_k k} = \sum_{i=1}^{k} \frac{1}{i\alpha_{i}}.
\end{align*}
This completes the induction.

Setting $\Lambda = E$ and $\tau = \emptyset$, the inequality \eqref{eq:IH-half-ineq} implies that the Poincar\'e constant of $Q$ is at least $\tp{\sum_{i=1}^{m}\frac{1}{i\alpha_{i}}}^{-1}$. 
This proves the first part of \Cref{thm:general-results}.

The second part of \Cref{thm:general-results}, concerning the log-Sobolev constant, follows similarly.
\end{proof}

\subsection{Local functional inequalities via transport flow}\label{sec:transport-flow}
We now prove \Cref{thm:HS-route-via-congestion} and \Cref{thm:log-Sob-HS-route-via-congestion}, establishing the local Poincar\'e and local log-Sobolev inequalities, respectively, by constructing transport flows with short paths and low congestion.

First, we establish the \HSRV{}.
Under the assumptions of bounded expected squared length and expected congestion, the \HSRV{} follows from an application of the Cauchy-Schwarz inequality, which also played a crucial role in the classic canonical path analysis by Jerrum and Sinclair \cite{jerrum2003counting}. 
However, unlike the canonical path method, which relies on the worst-case congestion and path length across all state pairs, our approach derives local functional inequalities based on the average congestion and path length between coupled pairs. 
This refinement allows for a significant improvement in the (local) Poincar\'e constant.


\begin{proof}[Proof of~\Cref{thm:HS-route-via-congestion}]
  Recall that $\mu$ is a distribution with support $\Omega \subseteq [q]^E$.
Let $f:\Omega \to \mathbb{R}$ be a function, let $X$ be drawn as $\*X \sim \mu$, and define $\*F = f(\*X)$.
  Fix an edge $e \in E$. Then, 
\begin{align*}
  \Var{\E{\*F \mid \*X_e}}
  &= \sum_{a,b \in [q]: a < b} \mu_e(a)\mu_e(b) \tp{\E[]{\*F \mid \*X_e = a} - \E[]{\*F \mid \*X_e = b}}^2.
\end{align*}
We can restrict the summation to those $a,b \in \Omega(\mu_e)\subseteq [q]$, where $\Omega(\mu_e)$ is the support of $\mu_e$, since others contribute zero to the sum.
Recall that $\Gamma_{e}^{a\to b}$ is the transport flow from $\mu^{e \gets a}$ to $\mu^{e \gets b}$,
and let $\Omega(\Gamma_e^{a\to b})$ denote its support. 
Define $\textsf{CP}_e^{a\to b}(x, y)\triangleq \{\gamma \in \Omega(\Gamma_e^{a\to b}):(x \mapsto y) \in \gamma\}$.

For a random path $\gamma = (\gamma_0, \gamma_1, \cdots, \gamma_{\ell(\gamma)}) \sim \Gamma_e^{a\to b}$, 
we have $\gamma_0\sim\mu^{e \gets a}$ and  $\gamma_{\ell(\gamma)}\sim\mu^{e \gets b}$.
Fixing $a, b \in \Omega(\mu_e)$ with $a < b$, by the linearity of expectation, we have
\begin{align*}
  \E[]{\*F \mid \*X_e = a} - \E[]{\*F \mid \*X_e = b}
  =&\E[\gamma \sim \Gamma_e^{a\to b}]{f(s(\gamma)) - f(t(\gamma))}
  =\E[\gamma \sim \Gamma_e^{a\to b}]{\sum_{0 \leq i < \ell(\gamma)} f(\gamma_{i}) - f(\gamma_{i+1})}\\
  =&\sum_{x,y \in \Omega:Q(x,y)>0}\sum_{\gamma \in \textsf{CP}_e^{a\to b}(x, y)}\Gamma_e^{a\to b}(\gamma)(f(x) - f(y))
\end{align*}
Applying the Cauchy-Schwarz inequality gives
\begin{align*}
  &\tp{\E[]{\*F \mid \*X_e = a} - \E[]{\*F \mid \*X_e = b}}^2\\
  \le\,&\tp{\sum_{x,y \in \Omega:Q(x,y)>0}\sum_{\gamma \in \textsf{CP}_e^{a\to b}(x, y)}\Gamma_e^{a\to b}(\gamma) \cdot \ell(\gamma)} \cdot \tp{\sum_{x,y \in \Omega:Q(x,y)>0}\sum_{\gamma \in \textsf{CP}_e^{a\to b}(x, y)}\frac{\Gamma_e^{a\to b}(\gamma)}{\ell(\gamma)}(f(x)-f(y))^2}\\
  =\,&\E[\gamma \sim \Gamma_e^{a\to b}]{\ell(\gamma)^2} \cdot \tp{\sum_{x,y \in \Omega:Q(x,y)>0}\E[\gamma \sim \Gamma_e^{a\to b}]{\frac{\*1[(x \mapsto y) \in \gamma]}{\ell(\gamma)}}(f(x)-f(y))^2}.
\end{align*}
By the assumption of bounded expected squared path length, we have 
$\E[\gamma \sim \Gamma_e^{a\to b}]{\ell(\gamma)^2} \le L$. 
Summing over all $a, b \in \Omega(\mu_e)$ with $a < b$ and over all $e\in E$ gives
\begin{align*}
  &\sum_{e \in E}\Var{\E{\*F \mid \*X_e}}\\
  \le\,&L\cdot\sum_{x,y \in \Omega:Q(x,y)>0}(f(x)-f(y))^2 \sum_{a,b \in [q]:a<b} \sum_{e \in E: a,b \in \Omega(\mu_e)} \mu_e(a)\mu_e(b) \cdot \E[\gamma \sim \Gamma_e^{a\to b}]{\frac{\*1[(x \mapsto y) \in \gamma]}{\ell(\gamma)}}
\end{align*}
By the assumption of bounded congestion, we have
\begin{align*}
  \sum_{e \in E}\Var{\E{\*F \mid \*X_e}} \leq q^2 \kappa L \cdot \sum_{x,y \in \Omega:Q(x,y)>0} \mu(x) Q(x,y) (f(x)-f(y))^2 = 2 q^2 \kappa L \cdot \+E_Q(f,f).
\end{align*}
This establishes the local Poincar\'e inequality.
\end{proof}

We also prove the local log-Sobolev inequality based on the existence of good transport flows.

\begin{proof}[Proof of \Cref{thm:log-Sob-HS-route-via-congestion}]
We continue using the notation from the proof of \Cref{thm:HS-route-via-congestion}, only this time, we analyze $\Ent[]{\E{\*F^2 \mid \*X_e}}$ instead of $\Var[]{\E{\*F \mid \*X_e}}$ and aim to establish the inequality in~\eqref{eq:def-HS-route-E}.

Fix an edge $e \in E$ with $|\Omega(\mu_e)| > 1$ (if no such edge $e$ exists, then \eqref{eq:def-HS-route-E} holds trivially).
Define $g_e: \Omega(\mu_e) \to \mathbb{R}$ by $g_e(c) = \E[]{\*F^2 \mid \*X_e = c}$. 
Consider the trivial  Markov chain $P_e$ that mixes to $\mu_e$ in one step,
with transition matrix $P_e(c,c') = \mu_e(c')$ for all $c,c' \in \Omega(\mu_e)$. 
The log-Sobolev constant of $P_e$ is at least $\frac{1-2\mu_e^*}{\log(1/\mu_e^* - 1)}$~\cite[Theorem A.1]{Diaconis1996logarithmic}, where $\mu_e^* = \min\{\mu_e(c)\mid c \in \Omega(\mu_e)\}$.
Since the function $y = \frac{1-2x}{\log(1/x - 1)}$ is increasing on $(0,1/2)$, we conclude that the log-Sobolev constant of $P_e$ is at least $\frac{1-2\phi}{\log(1/\phi - 1)}$, where $0 < \phi \leq \frac{1}{2}$ because $|\Omega(\mu_e)| > 1$.
Applying the log-Sobolev inequality, 
\begin{align}\label{eq:ent-sqrt-g-e}
  \Ent[]{\E[]{\*F^2 \mid \*X_e}} = \Ent[\mu_e]{g_e} &\leq \tp{\frac{\log(\frac{1}{\phi} - 1)}{1-2\phi}}\+E_{P_e}(\sqrt{g_e},\sqrt{g_e}).
\end{align}
Note that $\+E_{P_e}(\sqrt{g_e},\sqrt{g_e})=\Var[\mu_e]{\sqrt{g_e}}$.
By definition of the variance,
\begin{align*}
  \Var[\mu_e]{\sqrt{g_e}} = \sum_{a,b \in \Omega(\mu_e): a < b} \mu_e(a)\mu_e(b) \tp{\sqrt{\E[]{\*F^2 \mid \*X_e = a}} - \sqrt{\E[]{\*F^2 \mid \*X_e = b}}}^2.
\end{align*}
Let $\gamma = (\gamma_0, \gamma_1, \cdots, \gamma_{\ell(\gamma)})\sim \Gamma_{e}^{a\to b}$ be a path generated according to the transport flow $\Gamma_{e}^{a\to b}$ from $\mu_e^{e \gets a}$ to $\mu_e^{e \gets b}$.
Applying Jensen's inequality to the convex function $h(x,y) = (\sqrt{x} - \sqrt{y})^2$, 
\begin{align*}
  &\tp{\sqrt{\E{\*F^2 \mid \*X_e = a}} - \sqrt{\E{\*F^2 \mid \*X_e = b}}}^2 = \tp{\sqrt{\E[\gamma \sim \Gamma_{e}^{a\to b}]{f^2(s(\gamma))}} - \sqrt{\E[\gamma \sim \Gamma_{e}^{a\to b}]{f^2(t(\gamma))}}}^2\\
  &(\text{by convexity}) \quad \leq 
  \E[\gamma \sim \Gamma_{e}^{a\to b}]{\tp{f(s(\gamma)) - f(t(\gamma))}^2}
  = \E[\gamma \sim \Gamma_{e}^{a\to b}]{\tp{\sum_{0\leq i < \ell(\gamma)}\tp{f(\gamma_{i}) - f(\gamma_{i+1})}}^2}.
\end{align*}
Applying the Cauchy-Schwarz inequality gives
\begin{align*}
  &\E[\gamma \sim \Gamma_{e}^{a\to b}]{\tp{\sum_{0\leq i < \ell(\gamma)}\tp{f(\gamma_{i}) - f(\gamma_{i+1})}}^2} \leq  \E[\gamma \sim \Gamma_{e}^{a\to b}]{\ell(\gamma) \sum_{0\leq i < \ell(\gamma)} (f(\gamma_i)-f(\gamma_{i+1}))^2 }\\
  =\,& \E[\gamma \sim \Gamma_{e}^{a\to b}]{\ell(\gamma) \sum_{x,y \in \Omega: Q(x,y)>0} (f(x)-f(y))^2 \mathbf{1}[(x \mapsto y) \in \gamma] }\\
  =\,& \sum_{x,y \in \Omega: Q(x,y)>0} (f(x)-f(y))^2 \cdot \E[\gamma \sim \Gamma_{e}^{a\to b}]{\ell(\gamma) \mathbf{1}[(x \mapsto y) \in \gamma] }.
\end{align*}
Summing over all $a,b \in \Omega(\mu_e)$ with $a<b$ and over all $e \in E$, we obtain 
\begin{align*}
  \sum_{e \in E}\Var[\mu_e]{\sqrt{g_e}}
  \le\,& \sum_{x,y \in \Omega:\atop Q(x,y)>0}(f(x)-f(y))^2 \sum_{a,b \in [q]:\atop a<b}\sum_{e \in E:\atop a,b \in \Omega(\mu_e)}\mu_e(a)\mu_e(b) \E[\gamma \sim \Gamma_{e}^{a\to b}]{\ell(\gamma) \mathbf{1}[(x \mapsto y) \in \gamma]}\\
  \le\,& q^2\kappa\sum_{x,y \in \Omega: Q(x,y)>0}\mu(x)Q(x,y)(f(x)-f(y))^2 
  = 2q^2\kappa \cdot \+E_Q(f,f),
\end{align*}
where the last inequality follows by strong $\kappa$-expected congestion.
Combining with \eqref{eq:ent-sqrt-g-e}, we have
\[
  \sum_{e \in E}\Ent[]{\E[]{\*F^2 \mid \*X_e}}
  \le \tp{\frac{\log(\frac{1}{\phi} - 1)}{1-2\phi}}\sum_{e \in E}\Var[\mu_e]{\sqrt{g_e}}
  \le 2q^2\kappa \cdot \frac{\log(\frac{1}{\phi} - 1)}{1-2\phi} \cdot \+E_Q(f,f). \qedhere
\]
\end{proof}

\section{Application to the Jerrum-Sinclair Chain}\label{sec:app-Jerrume-Sinclair}
In this section, we apply the general framework from the previous section to the Jerrum-Sinclair chain on monomer-dimers,
utilizing the transport flow defined in \Cref{def:monomer-dimer-trans-path}.
Specifically:
\begin{itemize}
    \item In \Cref{sec:local-flipping-coupling},  we confirm the validity of the transport flow defined in \Cref{def:monomer-dimer-trans-path}.
    \item In \Cref{sec:length-overview} and \Cref{sec:expect-conges}, we bound the path length and congestion of this transport flow, respectively,  proving the corresponding bounds stated in \Cref{thm:effective-trans-flow}.
    Next, in \Cref{sec:length-and-congestion-analysis-for-conditional-distributions}, we extend these bounds to chains under arbitrary pinning, thereby implying \Cref{lem:local-functional-inequalities}.
    \item Finally, in \Cref{sec:concave-Dirichlet-forms}, we verify the concavity of the Dirichlet forms of the Jerrum-Sinclair chain and  prove the marginal bounds for the monomer-dimer distributions.
\end{itemize}
Throughout this section, we adopt the setup from \Cref{sec:app-JS-chain}, particularly that of \Cref{thm:effective-trans-flow}.

\newcommand{\ptplus}{{\partial}\mspace{-3mu}}

Let $G=(V,E)$ be a simple, undirected graph with $n$ vertices, $m$ edges, and maximum degree $\Delta$.
Consider the monomer-dimer distribution $\mu$   on $G$ with edge weight $\lambda$, and define $\overline{\lambda} = \max\set{1,\lambda}$.
For each edge $e\in E$, let $\+C_e$ denote the local flipping coupling  of $\mu^{e \gets 0}$ and $\mu^{e \gets 1}$ in \Cref{def:new-coupling-2}; 
and let $\Gamma_e$ denote the transport flow from $\mu^{e \gets 0}$ to $\mu^{e \gets 1}$ in~\Cref{def:monomer-dimer-trans-path}. 

For any pair of configurations $x,y\in\Omega(\mu)$, where $\Omega(\cdot)$ denotes the support of a distribution, let $\gamma^{x,y}$ be the Jerrum-Sinclair canonical path from $x$ to $y$, as specified in \Cref{def:monomer-dimer-trans-path}. 
Then, for a random pair $(X,Y)\sim\+C_e$ generated according to the local flipping coupling $\+C_e$, the canonical path $\gamma^{X,Y}$ from $X$ to $Y$ follows the distribution of the transport flow $\Gamma_e$.

Additionally, we introduce some notation.
For any subset $S \subseteq E$  of edges,
we define the \emph{inclusive boundary} $\ol{\partial}S$ and the \emph{boundary} $\partial S$  as follows:
\begin{align*}
  \ol{\partial} S \triangleq \set{e \in E \mid \exists f \in S, e \cap f \neq \emptyset } \quad \text{and} \quad \partial S \triangleq \ol{\partial} S \setminus S.
\end{align*}
For any configurations $x,y\in\{0,1\}^E$, we denote by $x\oplus y$  the set of edges where $x$ and $y$ differ:
\[
x\oplus y\triangleq \{f\in E\mid x_f\neq y_f\}.
\]
As discussed in \Cref{sec:app-JS-chain}, we interpret each configuration $x\in\{0,1\}^E$ equivalently as the set of edges it indicates, i.e., $\{e\in E\mid x(e)=1\}$.
Thus, set operations, including union~($\cup$), intersection~($\cap$), set difference~($\setminus$), symmetric difference~($\triangle$), and cardinality~($|\cdot|$), can be applied to configurations in $\{0,1\}^E$ without ambiguity.

\subsection{Validity of the local flipping coupling}\label{sec:local-flipping-coupling}
The well-definedness of the transport flow follows from the fact that the local flipping coupling, as defined in \Cref{def:new-coupling-2}, provides a valid coupling of $\mu^{e \gets 0}$ and $\mu^{e \gets 1}$.


\begin{proposition}\label{lem:explicit-formula}
  Fix an edge $e\in E$. The local flipping coupling $\+C_e$ defined in \Cref{def:new-coupling-2} is a valid coupling of $\mu^{e \gets 0}$ and $\mu^{e \gets 1}$. Furthermore, for $x,y \in \Omega(\mu)$, let $B=x\oplus y$ be the set of disagreement edges. 
  \begin{enumerate}
        \item\label{case:path-cycle} If $B$ forms a path or a cycle of even length and $x_e=0,y_e=1$, the probability of $(x,y)$ in $\+C_e$ is 
        \begin{align}\label{eq:explicit-formula}
            \Pr[(X,Y) \sim \+C_e]{(X,Y) = (x,y)} = \mu_{{\overline{\partial}} B}^{e \gets 0}(x_{{\overline{\partial}} B}) \cdot\mu^{e \gets 1}(y).
        \end{align}
        \item Otherwise, $\Pr[(X,Y) \sim \+C_e]{(X,Y) = (x,y)} = 0$.
    \end{enumerate}
\end{proposition}
\begin{proof}
  We first verify \eqref{eq:explicit-formula} for Case~1, noting that Case 2 is obvious.
  Let $X,Y,Z$  be the random configurations  generated in \Cref{def:new-coupling-2} and let $\boldsymbol{B}=X\oplus Y$.
  It holds that $X_f=Z_f \neq Y_f$ for any $f \in \boldsymbol{B}$, and $X_f = Y_f = Z_f = 0$ for any $f \in \partial \boldsymbol{B}$.
  Then, for any $x \in \Omega(\mu^{e \gets 0})$ and $y \in \Omega(\mu^{e \gets 1})$, the event $(X,Y)=(x,y)$ is equivalent to that $Y = y$ and $Z_{{\overline{\partial}} B} = x_{{\overline{\partial}} B}$, where $B=x\oplus y$. Therefore,
  \begin{align*}
      \Pr[(X,Y) \sim \+C_e]{(X,Y) = (x,y)} = \Pr[Z \sim \mu^{e \gets 0}]{Z_{{\overline{\partial}} B} = x_{{\overline{\partial}} B}} \cdot \Pr[Y \sim \mu^{e \gets 1}]{Y=y}  = \mu^{e \gets 0}_{{\overline{\partial}} B}(x_{{\overline{\partial}} B}) \cdot\mu^{e \gets 1}(y).
  \end{align*}
  Consider a sample $(X,Y) \sim \+C_e$ from the coupling $\+C_e$. The marginal distribution of $Y$ is clearly $\mu^{e \gets 1}$. To find the marginal distribution of $X$, we note 
  \begin{align*}
      \mu^{e \gets 0}_{{\overline{\partial}} B}(x_{{\overline{\partial}} B}) \cdot\mu^{e \gets 1}(y) &= \mu^{e \gets 0}_{{\overline{\partial}} B}(x_{{\overline{\partial}} B}) \cdot\mu^{e \gets 1}_{{\overline{\partial}} B}(y_{{\overline{\partial}} B}) \cdot\mu^{y_{{\overline{\partial}} B}}_{E \setminus {\overline{\partial}} B}(y_{E \setminus {\overline{\partial}} B})\\
      &\overset{(\star)}{=} \mu^{e \gets 0}_{{\overline{\partial}} B}(x_{{\overline{\partial}} B}) \cdot \mu^{e \gets 1}_{{\overline{\partial}} B}(y_{{\overline{\partial}} B}) \cdot \mu^{x_{{\overline{\partial}} B}}_{E \setminus {\overline{\partial}} B}(x_{E \setminus {\overline{\partial}} B}) = \mu^{e \gets 0}(x) \cdot \mu^{e \gets 1}_{{\overline{\partial}} B}(y_{{\overline{\partial}} B}),
  \end{align*}
  where the equality $(\star)$ follows from the conditional independence property of $\mu$, as $x_{\ptplus B} = y_{\ptplus B} = \*0_{\ptplus B}$ and $x_{E \setminus \ol{\partial} B} = y_{E \setminus \ol{\partial} B}$. 
  Therefore, we can consider a symmetric process $\+C_e'$ for generating $(X,Y)$, which first samples $X \sim \mu^{e \gets 0}$ and $Z \sim \mu^{e \gets 1}$,  identifies the path or cycle $B\subseteq X\oplus Z$ containing $e$, and sets $Y = Z_B \uplus X_{E \setminus B}$. The above calculation shows that $\+C_e'$ is identically distributed as $\+C_e$. Finally, it is easy to verify from the construction of $\+C_e'$ that $X$ follows the law of $\mu^{e \gets 0}$, as required.
\end{proof}
Note that \eqref{eq:explicit-formula} can be re-expressed as: 
$$\mu_e(0)\cdot\mu_e(1)\cdot\Pr[(X,Y) \sim \+C_e]{(X,Y) = (x,y)} = \mu_{{\overline{\partial}} B}(x_{{\overline{\partial}} B})\cdot \mu(y).$$ 
Observe that the quantity $\mu_{{\overline{\partial}} B}(x_{{\overline{\partial}} B})\cdot \mu(y)$ on the right-hand-side is independent of the fixed edge~$e$. This symmetry plays a vital role in the decoupling process (see~\Cref{lem:pre-decoupling}), which will be used later in the congestion analysis of the transport flow.


\subsection{Moment bounds on the coupling discrepancy}\label{sec:length-overview}
The following lemma bounds the 2nd moment of the discrepancy of the local flipping coupling $\+C_e$,
implying the expected squared path length of the transport flow, as stated in \Cref{thm:effective-trans-flow}(\ref{thm:trans-flow-length}).
\begin{lemma} \label{cor:expected-squared-length}
Fix an edge $e=\{u,v\} \in E$.
  The transport flow $\Gamma_e$ defined in~\Cref{def:monomer-dimer-trans-path} satisfies 
  \begin{align*}
    \E[\gamma\sim \Gamma_e]{\ell(\gamma)^2}\leq \E[(X,Y)\sim \+C_e]{|X\oplus Y|^2} = O(\ol{\lambda}\Delta  \log^2(1+\ol{\lambda})).
  \end{align*}
\end{lemma}


\begin{proof}

Let $(X,Y) \sim \+C_e$ be generated by the local flipping coupling $\+C_e$ according to \Cref{def:new-coupling-2},
let $B = X\oplus Y$, and let $\gamma\sim\Gamma_e$ be generated as in \Cref{def:monomer-dimer-trans-path}. 
It holds $\ell(\gamma)\le  |X\oplus Y|=|B|$. 


The proof proceeds in two cases: when $B$ is a path or a cycle of even length.

\paragraph{Path case.} 
When $B$ forms a path, removing the edge $e=\set{u,v}$ partitions $B$ into two disjoint paths $P_v$ and $P_u$, where $P_v$(resp.~$P_u$) is the path connected to $v$(resp.~$u$).  Together, it satisfies $B = P_v \cup \{e\} \cup P_u$. Note that either $P_v$ or $P_u$ may be empty (i.e., contain no edges). Now, consider a fixed (possibly empty) path $P$ that starts from $v$ and does not include $u$. We claim that, conditional on $P_v = P$, the length of $P_u$ is stochastically dominated by $2L$, where $L$ is a geometric random variable with success probability $q=\frac{2}{1+\sqrt{1+\lambda\Delta}}$. 
Formally, let  $P_v, P_u$ be constructed from $(X,Y)\sim \+C_e$ as described above, given that $B$ is a path.
The following lemma holds: 
  \begin{lemma}\label{lem:path-length-geo-bound}
    For any edge $e=\set{u,v}\in E$ and  any path $P$ with $v$ as one of its endpoints (or $P=\emptyset$), if 
    $P\cup \set{e}$ does not form a cycle (which means $\Pr{P_v = P \mid B\text{ is a path}} > 0$), 
    then there exists a coupling between a geometric random variable $L\sim\-{Geo}(q)$, where $q\triangleq \frac{2}{1+\sqrt{1+\lambda\Delta}}$, and $(X,Y)\sim \+C_e$, such that
    \begin{align*}
      \Pr{\abs{P_u} \le 2L\mid B\text{ is a path} \wedge P_v=P}=1. 
    \end{align*}
  \end{lemma}


By \Cref{lem:path-length-geo-bound}, the law of total expectation, and the 2nd moment of a geometric random variable, 
\begin{align*}
  \E{\abs{P_u}^2 \mid B\text{ is a path}}  
  = \E{\E[]{\abs{P_u}^2 \mid P_v} \mid B\text{ is a path}}
  \le \E{4L^2} =  O(q^{-2}).
\end{align*}
%
By symmetry, the same bound applies to $\E{|P_v|^2 \mid B\text{ is a path}}$ as well.

Next, since  $|B| = |P_u| + |P_v| + 1=\langle(|P_u|,|P_v|,1),(1,1,1)\rangle$, applying Cauchy-Schwarz, we have
$|B|^2 \leq 3(|P_u|^2+|P_v|^2+1)$,
which leads to 
\[ \E{|B|^2\cdot \*1[B\text{ is a path}]} \leq \E{|B|^2 \mid B\text{ is a path}}  = O(q^{-2}) = O(1 + \lambda\Delta). \]

\paragraph{Cycle case.}
Let $k = \theta q^{-1}$, where $\theta > 0$  is a parameter to be determined later.
We decompose the second moment based on cycle size as follows:
\begin{align*}
  \E{\abs{B}^2 \cdot \*1[B \text{ is a cycle}]}
  &= \E{\abs{B}^2 \cdot \*1[B \text{ is a cycle}] \cdot \*1[\abs{B} \leq k]} + \E{\abs{B}^2 \cdot \*1[B \text{ is a cycle}] \cdot \*1[\abs{B} \geq k]} \\
  &\leq \theta^2q^{-2} + \E{\abs{B}^2 \cdot \*1[B \text{ is a cycle}] \cdot \*1[\abs{B} \geq k]}.
\end{align*}
Thus, we focus on the second term, which accounts for large cycles.
For each even-length cycle  $C$ containing $e=\{u,v\}$, let $P$ be the path obtained by removing the unique edge $f\in C$ such that $f \ni v$ and $f\neq e$. 
We claim 
\begin{align} \label{eq:prob-cycle-path}
  \Pr{B = C} = \lambda \cdot \Pr{B = P}=\lambda \cdot \Pr{B = P\land P_v=\emptyset},
\end{align}
which defines an injection $C\mapsto P$ with $|C|=|P|+1$ from cycles to paths satisfying \eqref{eq:prob-cycle-path}, implying
\begin{align*}
  \E{\abs{B}^2 \cdot \*1[B \text{ is a cycle}] \cdot \*1[\abs{B} \geq k]}
  \le\,& \lambda \E{\*1[B \text{ is a path}] \cdot (\abs{P_u}+2)^2 \cdot \*1[P_v = \emptyset] \cdot \*1[\abs{P_u} + 2 \geq k]} \\ 
  \leq\,& \lambda \E{(\abs{P_u}+2)^2 \cdot \*1[\abs{P_u} + 2 \geq k] \mid B\text{ is a path} \wedge P_v = \emptyset} \\
  (\text{by \Cref{lem:path-length-geo-bound}})\quad 
  \leq\,& \lambda \E{(2L+2)^2\cdot \*1[2L+2\ge k]}\\
  \leq\,& O(\lambda \exp(-\theta/2) q^{-2} (\theta^2+1)).
\end{align*}
Choosing $\theta = 10\log \ol{\lambda}$ yields
\begin{align*}
  \E{\abs{B}^2 \cdot \*1[B \text{ is a cycle}]} = O(q^{-2}(1+\log^2\ol{\lambda})) = O((1+\lambda\Delta)(1+\log^2\ol{\lambda})).
\end{align*}

Now, it remains to show \eqref{eq:prob-cycle-path}, i.e., $\Pr{B = C}=\lambda \cdot \Pr{B=P}$ for any even cycle $C$ containing $e=\{u,v\}$ and the path $P$ obtained by removing the unique edge $f\in C\setminus\{e\}$ incident to vertex~$v$. 
This follows from \Cref{lem:explicit-formula} and the fact that  ${\overline{\partial}} C = {\overline{\partial}} P$. 
Specifically, by~\Cref{lem:explicit-formula}, 
\begin{align}\label{eq:mid-1}
\Pr{B = C} = \mu^{e \gets 0}_{{\overline{\partial}} C}(x_{{\overline{\partial}} C}) \cdot \mu^{e \gets 1}_{{\overline{\partial}} C}(y_{{\overline{\partial}} C}),
\end{align}
where $x\in \Omega(\mu^{e \gets 0})$ and $y \in \Omega(\mu^{e \gets 1})$ correspond to the two matchings with difference set $B = C$. 
Note that given $x_e=0, y_e=1$, there is only one way to alternate $x$ to $y$, which uniquely determines the configurations $x_{{\overline{\partial}} C}$ and $y_{{\overline{\partial}} C}$ given their difference $C$. Thus, \eqref{eq:mid-1} is well-defined.

Let $f \in C \setminus \{e\}$ be the edge removed from $C$  to obtain $P$, and define $x^{-f}$ (resp.~$y^{-f}$) as the configuration obtained by setting $x_f = 0$ (resp.~$y_f = 0$). 
Similarly, by~\Cref{lem:explicit-formula}, we have
\begin{align}\label{eq:mid-2}
\Pr{B = P} = \mu_{{\overline{\partial}} P}^{e \gets 0}(x_{{\overline{\partial}} P}^{-f})\cdot \mu_{{\overline{\partial}} P}^{e \gets 1}(y_{{\overline{\partial}} P}^{-f}).
\end{align} 
Since $C$ is the difference between $x$ and $y$, we have $x_{\partial C} = y_{\partial C} = \*0_{\partial C}$. It then  follows that
\begin{align}\label{eq:mid-3}
\mu^{e \gets 0}_{{\overline{\partial}} C}(x_{{\overline{\partial}} C}) \cdot \mu^{e \gets 1}_{{\overline{\partial}} C}(y_{{\overline{\partial}} C}) = \lambda \cdot \mu^{e \gets 0}_{{\overline{\partial}} C}(x_{{\overline{\partial}} C}^{-f}) \cdot \mu^{e \gets 1}_{{\overline{\partial}} C}(y_{{\overline{\partial}} C}^{-f}).
\end{align}
The desired equation~\eqref{eq:prob-cycle-path} then follows from \eqref{eq:mid-1}, \eqref{eq:mid-2}, \eqref{eq:mid-3}, and the fact that ${\overline{\partial}} C = {\overline{\partial}} P$.
\end{proof}


\subsubsection{Geometric decay of the discrepancy path}
In the following, we provide a proof of~\Cref{lem:path-length-geo-bound}.

\begin{proof}[Proof of~\Cref{lem:path-length-geo-bound}]
   Recall that we use ${\overline{\partial}} P_v$ to denote the inclusive boundary of the path $P_v$, which includes all edges incident to the path, as well as the edges that make up the path itself. 
   In the case where $P_v = \emptyset$,  we abuse notation and define ${\overline{\partial}} P_v = \{f \in E \mid v \in f\}$.
    Conditioned on $P_v$, the configurations $(X,Y)$ on the inclusive boundary ${\overline{\partial}} P_v$ are uniquely determined, and we have $X_{\partial P_v \setminus \{e\}} = Y_{\partial P_v \setminus \{e\}} = \*0_{\partial P_v \setminus \{e\}}$.  
    Therefore, for any $(x,y)\in\Omega(\+C_e)$ whose difference set  $B=x\oplus y$ forms  a path containing $P_v$, we have
    \begin{align*}
        &\Pr[(X,Y) \sim \+C_e]{(X,Y) = (x,y) \mid {X\oplus Y\text{ is a path } \wedge} P_v} \\
        = &\frac{\Pr[(X,Y) \sim \+C_e]{(X,Y) = (x,y)}}{\Pr[(X,Y) \sim \+C_e]{{X\oplus Y\text{ is a path }\wedge}X\oplus Y\supseteq P_v\wedge X_{\partial P_v \setminus \{e\}} = Y_{\partial P_v \setminus \{e\}} = \*0_{\partial P_v \setminus \{e\}}}}\\
        {(\text{\small by \Cref{lem:explicit-formula}})\atop(\text{\small by \Cref{def:new-coupling-2}})}\quad =&\frac{\mu_{{\overline{\partial}} B}^{e \gets 0}(x_{{\overline{\partial}} B})\cdot \mu^{e \gets 1}_{ {\overline{\partial}} B}(y_{ {\overline{\partial}} B}) \cdot\mu^{x_{{\overline{\partial}} B}}_{E \setminus {\overline{\partial}} B}(x_{E \setminus{\overline{\partial}} B})}{\mu^{e \gets 0}_{ {\overline{\partial}} P_v}(x_{ {\overline{\partial}} P_v}) \cdot\mu^{e \gets 1}_{ {\overline{\partial}} P_v}(y_{{\overline{\partial}} P_v})}\\
   \text{(by $P_v \subseteq B$)}\quad     =& \mu^{x_{{\overline{\partial}} P_v}}_{{\overline{\partial}} B}(x_{ {\overline{\partial}} B}) \cdot\mu^{y_{ {\overline{\partial}} P_v}}_{ {\overline{\partial}} B}(y_{{\overline{\partial}} B}) \cdot\mu^{x_{ {\overline{\partial}} B}}_{E \setminus  {\overline{\partial}} B}(x_{E \setminus {\overline{\partial}} B})\\
   \text{(by $e \in {\overline{\partial}} P_v$)}\quad     =&\mu^{\*0_{\partial P_v \setminus \{e\}} \wedge (e \gets 0)}_{ {\overline{\partial}} B \setminus \ol{\partial} P_v }(x_{{\overline{\partial}} B \setminus 
\ol{\partial} P_v }) \cdot\mu^{\*0_{\partial P_v \setminus \{e\}} \wedge (e \gets 1)}_{ {\overline{\partial}B \setminus \ol{\partial} P_v } }(y_{{\overline{\partial}} B \setminus \ol{\partial} P_v }) \cdot\mu^{x_{{\overline{\partial}} B}}_{E \setminus {\overline{\partial}} B}(x_{E \setminus {\overline{\partial}} B}).
    \end{align*}
 Therefore,  conditioning on  $P_v$ is equivalent to considering the monomer-dimer distribution on a subgraph of $G$  obtained by removing all edges in ${\overline{\partial}} P_v \setminus \{e\}$.
    Formally, the distribution of $(X,Y) \sim \+C_e$ conditioned on $P_v$ and projected onto $E \setminus ({\overline{\partial}} P_v \setminus \{e\})$  follows the law of the local flipping coupling $\widetilde{\+C_e}$ of the distributions  $\widetilde{\mu}^{e \gets 0}$ and $\widetilde{\mu}^{e \gets 1}$, where $\widetilde{\mu}$ is the monomer-dimer distribution on the graph $\widetilde{G} \triangleq G\setminus ({\overline{\partial}} P_v \setminus \{e\})$. 
    Let $\widetilde{X} \sim \widetilde{\mu}^{e \gets 0}$ and $\widetilde{Y} \sim \widetilde{\mu}^{e \gets 1}$ be independent samples,
    and define $\widetilde{B}=\widetilde{X}\oplus\widetilde{Y}$ as the set of edges where $\widetilde{X}$ and $\widetilde{Y}$ disagree.
    The following process traverses $\widetilde{B}$ by revealing the disagreement edges between $\widetilde{X}$ and $\widetilde{Y}$ one
     by one:
    \begin{enumerate}
    \item \emph{Initialization}:  Sample $\widetilde{X} \sim \widetilde{\mu}^{e \gets 0}$ and $\widetilde{Y} \sim \widetilde{\mu}^{e \gets 1}$ independently. Initialize $w^\star \gets u$, where $u$  is the   vertex (other than $v$) incident to the edge $e$.
    Initially, only the disagreement of $\widetilde{X}$ and $\widetilde{Y}$ at the edge $e$ are revealed.
    \item \emph{Self-avoiding walk}: 
    Reveal the partial configurations of $X$ and $Y$ on edges incident to $w^\star$. If a new disagreement $\{w^\star,w\} \in E$ between $\widetilde{X}$ and $\widetilde{Y}$ is encountered, update $w^\star \gets w$ (note that there can be at most one new disagreement). If no new disagreement is found, the process terminates.  The path of disagreements, $\widetilde{B}$, is given precisely by the trail of the vertex $w^\star$.
    \end{enumerate}
    
    Note that $e$ is a hanging edge in the new graph $\widetilde{G}$.
    Thus, intuitively, the above process should be exactly the same as the coupling used to construct the ``path tree'' (i.e., Godsil's self-avoiding walk tree for the monomer-dimer model, see~\cite{godsil1993algebratic,bayati2007simple}) rooted at $v$ on $\widetilde{G}$.
    
    We claim that at any point in the process,  it terminates within the next two steps with probability $q = \Omega(1/\sqrt{1+\lambda\Delta})$.
    Specifically, conditioned on any $w^\star$, with probability $q$, the process either terminates at the current $w^\star$, or $w^\star$ moves to a new neighbor $w$ and then terminates.
    This is consistent with the well-known decay of correlation  property for the monomer-dimer distributions in~\cite{bayati2007simple}. 
    For completeness, we include a proof of this result in our setting.

    Suppose $w^\star$ is the current disagreement vertex, and let $e'$ be the last revealed disagreement edge incident to $w^\star$. 
    Without loss of generality, assume $\widetilde{X}_{e'} = 1$ and $\widetilde{Y}_{e'} = 0$; the other case follows by symmetry. 
    By self-reducibility, remove all edges (including $e'$) revealed during the process and consider the monomer-dimer model on the remaining subgraph $G'$. 
    Let $N(w^\star)$ denote the set of neighbors of $w^\star$ in $G'$, and let $\mu'$ be the monomer-dimer distribution on $G'$.
    If the process does not terminate within next two steps, there must exist a neighbor $w \in N(w^\star)$ such that:
    \begin{enumerate}
        \item In the first step, $\widetilde{Y}_{\{w^\star, w\}} = 1$, which occurs with probability $\mu'(\{w,w^\star\} \text{ is occupied})$;
        \item In the next step, $w$ becomes saturated in $\widetilde{X}$, which occurs with probability 
        $$\mu'(\text{$w$ is saturated} \mid \text{$w^\star$ is not saturated}),$$ 
        since given $\widetilde{X}_e = 1$, no other edge in $\widetilde{X}$ is incident to $w^\star$, i.e., $w^\star$ is not saturated by $\widetilde{X}$ in $G'$;
    \end{enumerate}
    where the above probabilities are calculated by the \emph{principle of deferred decisions}.
    
    Let $\+B$ denote the event that the process does not terminate within the next two steps starting from $w^\star$.
    Since $\widetilde{X}$ and $\widetilde{Y}$ are independent, we have 
    \begin{align*}
      \Pr{\+B} = 
        \sum_{w \in N(w^\star)} \mu'(\text{$\{w,w^\star\}$ is occupied}) \cdot \mu'(\text{$w$ is saturated} \mid \text{$w^\star$ is not saturated}).
    \end{align*}
    Define $p_w \triangleq \mu'(\text{$w$ is not saturated} \mid \text{$w^\star$ is not saturated})$. 
    The probability of the edge $\{w,w^\star\}$ being occupied in $\mu'$ is given by
    \begin{align*}
    \mu'(\text{$\{w,w^\star\}$ is occupied}) =\lambda\cdot\mu'(w \text{ and } w^\star \text{ are not saturated})= \frac{\lambda p_w}{1+\lambda \sum_{w' \in N(w^\star)} p_{w'}}.
    \end{align*}
    To verify this equation, we multiply both the numerator and denominator by the probability that $w^\star$ is not saturated and observe: 
    (1)~the event that $w^\star$ is saturated is equivalent to the event that no edge $\{w,w^\star\}$ is occupied;
    and (2)~in any monomer-dimer distribution, it always holds
    $\Pr[]{\{w,w^\star\} \text{ is occupied}}=\lambda\cdot  \Pr[]{w \text{ and } w^\star \text{ are not saturated}}$.
    
    Define $S \triangleq \sum_{w \in N(w^\star)} p_w$. The (two-step) non-termination probability  $\Pr{\+B}$ satisfies
    \begin{align*}
        \Pr{\+B} = \frac{\sum_{w \in N(w^\star)} \lambda p_w (1-p_w)}{1+\lambda \sum_{w \in N(w^\star)} p_w} =  \frac{\lambda S - \lambda \sum_{w \in N(w^\star)} p_w^2}{1+\lambda S} \leq \frac{\lambda (S-S^2/\Delta)}{1+\lambda S},
    \end{align*}
    where the last inequality follows the Cauchy-Schwarz. Setting $T = S/\Delta$, we have
    \begin{align*}
        \Pr{\+B} \le \max_{T \in [0,1]} \frac{\lambda \Delta \cdot (T-T^2)}{1+\lambda \Delta \cdot T} \le 1-\frac{2}{1+\sqrt{1+\lambda \Delta}}.
    \end{align*}
    Thus, the discrepancy $|\widetilde{B}|$ is stochastically dominated by $2L$, where $L\sim\-{Geo}(q)$ for $q = \frac{2}{1+\sqrt{1+\lambda \Delta}}$. 
    
    To construct a coupling between $L$ and $(X,Y)\sim \+C_e$, we perform a Bernoulli trial with success probability $q$ every two steps and terminate the revealing process when the trial succeeds.
\end{proof}

\subsubsection{$p$-th moment bound on the one-sided discrepancy}
Beyond the 2nd moment bound in \Cref{cor:expected-squared-length}, we also analyze the $p$-th moment of the discrepancy, conditioned on one coordinate of the coupling being fixed. This is called the \emph{one-sided discrepancy}, which plays a key role in the analysis of congestion.
\begin{lemma}\label{lem:one-sided-CI}
  Fix an edge $e=\{u,v\} \in E$ and any constant $p\ge 0$. For any $x\in \Omega(\mu^{e \gets 0})$ and $y\in \Omega(\mu^{e \gets 1})$, 
  \begin{align*}
      \E[(X,Y)\sim {\+C}_e]{|{X\oplus Y}|^p\mid X=x} &=O_p(\ol{\lambda}^p \Delta^p);\\
      \E[(X,Y)\sim {\+C}_e]{|X\oplus Y|^p\mid Y=y}   &=O_p(\ol{\lambda}^p \Delta^p).
  \end{align*} 
\end{lemma}

In \Cref{lem:one-sided-CI},  we fix the (one-side) outcome of the coupling $\+C_e$ to be $X=x$ (resp.~$Y=y$) and analyze the conditional distribution of $X$ given $X = x$ (resp.~$Y=y$). 
The lemma asserts that, even under this conditional distribution, the $p$-th moment of the discrepancy length remains bounded. 
By the law of total expectation, \Cref{lem:one-sided-CI} immediately implies that $  \E[(X,Y)\sim {\+C}_e]{|X\oplus Y|^p}=O_p(\ol{\lambda}^p \Delta^p)$. 
However, \Cref{cor:expected-squared-length} establishes a sharper bound for $p = 2$. 
To prove \Cref{lem:one-sided-CI}, we leverage the definition of the local flipping coupling (\Cref{def:new-coupling-2}) and employ a local revealing process.  

\begin{proof}
We focus on the case where the one-sided outcome $Y$ of the coupling $\+C_e$ is fixed, while the case with $X$ fixed follows by symmetry.
Fix a configuration $y \in \Omega(\mu)$ with $y_e = 1$.
Let $(X,y)$ be sampled from the local flipping coupling $\+C_e$ given $Y=y$, let  $B=X\oplus y$ be the disagreement set. 

The following process traverses $B$ by revealing  the disagreement edges one by one:
\begin{enumerate}
    \item Sample $X \sim \mu^{e \gets 0}$ and initialize a stack $S$ (which supports \emph{push} and \emph{pop} operations) by pushing $v$, then $u$. Initially, the states of all edges in $X$ and $y$ are unrevealed except for the edge~$e$. 
    \item While $S\neq\emptyset$, pop the top vertex $w^\star$ from $S$ and reveal the states of all edges in $X$ and $y$ incident to $w^\star$.  If a new disagreement  $\{w^\star,w\} \in E$ between $X$ and $y$ is found,  push $w$ into~$S$. 
\end{enumerate}
Let $w_0 = u, w_1, \ldots, w_{k-1}, w_k =v, w_{k+1}, \ldots , w_\ell$ be the sequence of vertices $w^\star$ popped in the second step of the process, listed in chronological order. 
If $v$ appears multiple times in the sequence, define $k$ as the smallest index for which $w_k = v$. 
The disagreement set $B$ consists of: 
the initially selected edge $e$, 
the edges  $\{w_0,w_1\}, \{w_1,w_2\}, \ldots, \{w_{k-2},w_{k-1}\}$ forming the path from $u$ to $w_{k-1}$,
and the edges $\{w_k,w_{k+1}\}, \ldots, \{w_{\ell-1},w_\ell\}$ from $v$ onward. 
Define $P_u$ as the path from $w_0$ to $w_{k-1}$ and $P_v$ as the path from $w_k$ to $w_\ell$. 
The size of discrepancy satisfies $\abs{B}\le \abs{P_u} + \abs{P_v}$. 
We will show that $\abs{P_u}$ and $\abs{P_v}$ are each stochastically dominated by $2L$, where $L$ follows a geometric distribution with success probability $q = \frac{1}{1+\lambda \Delta}$. 
It suffices to prove this for $P_u$, and for $P_v$, the same argument applies. 

Specifically, we show that for any even $i$, the probability that $P_u$ does not terminate at $w_i$ is at most $\frac{\lambda \Delta}{1+\lambda \Delta}$. 
This follows from the \emph{principle of deferred decisions}:
Observe that the sequence of edges $e_0 = e, e_1 = \{w_0,w_1\}, e_2= \{w_1,w_2\}, \ldots, e_i = \{w_{i-1},w_i\}$ alternates between being present in $y$ and $X$.  
In particular, for even $i$, the edge $e_i$ belongs to $y$, meaning that for $P_u$ to continue past $w_i$, a new disagreement $\{w_i,w\} \in E$ between $X$ and $y$ must be generated. 
The probability of this occurring is at most the probability that $w_i$ is saturated in $X$, which can be upper bounded by $\frac{\lambda \Delta}{1+ \lambda \Delta}$.
Applying this argument iteratively gives the desired stochastic domination by a geometric random variable. 
Finally, using the standard moment bound for geometric distributions, $\E{L^p} = O_p(q^{-p})$, we conclude that $\E{|P_v|^p},\E{|P_u|^p}\in O_p(q^{-p})=O_p(\ol{\lambda}^p \Delta^p)$.
Thus, we obtain
\begin{align*}
\E{|B|^p}
&\le \E{(|P_v|+|P_u|)^p} \le \E{(2\cdot \max\{|P_v|,|P_u|\})^p}= 2^p\cdot \E{\max\{|P_v|^p,|P_u|^p\}}\\
&\le 2^p\cdot \E{|P_v|^p+|P_u|^p}
= O_p(\ol{\lambda}^p \Delta^p).    
\end{align*}
This proves the bound for the case where $Y=y$.
The case where $X=x$ follows by symmetry.
\end{proof}

\subsection{Bounded congestion of the transport flow} \label{sec:expect-conges}

Next, we analyze the congestion of the transport flow defined in~\Cref{def:monomer-dimer-trans-path} using an approach based on \DCP{} and measure-preserved flow encoding, proving \Cref{thm:effective-trans-flow}(\ref{thm:trans-flow-congestion})-(\ref{thm:trans-flow-strong-congestion}).

Let $Q$ denote the Jerrum-Sinclair chain for the monomer-dimer distribution $\mu$.
For any pair of configurations $X,Y\in\Omega(\mu)$ such that their difference $X\oplus Y\triangleq \{f\in E\mid X_f\neq Y_f\}$ forms a path or a cycle, let $\gamma^{X,Y}$ denote the unique canonical path from $X$ to $Y$ constructed according to \Cref{def:monomer-dimer-trans-path}.
It is obvious that if $(X,Y)\sim\+C_e$ is sampled from the local flipping coupling $\+C_e$, then $\gamma^{X,Y}\sim\Gamma_e$, following the law of the transport flow $\Gamma_e$ from $\mu^{e \gets 0}$ to $\mu^{e \gets 1}$.


\begin{lemma} \label{lem:expected-congestion}
The transport flow $(\Gamma_e)_{e \in E}$ defined in~\Cref{def:monomer-dimer-trans-path} satisfies the following:
\begin{enumerate}
    \item \emph{($O(\ol{\lambda}^{2}m)$-expected congestion)} For any transition $(\alpha \mapsto \beta)$ with $Q(\alpha, \beta) > 0$, 
  \begin{align*}
    \sum_{e\in E} \mu_e(0)\mu_e(1)\E[(X,Y)\sim \+C_e]{\frac{\*1[(\alpha \mapsto \beta) \in \gamma^{X,Y}]}{\ell(\gamma^{X,Y})}} 
    &\leq O(\ol{\lambda}^{2} m) \cdot \mu(\alpha)Q(\alpha, \beta);
  \end{align*}
    \item \emph{(strong $O(\ol{\lambda}^{4}\Delta^2 m)$-expected congestion)} For any transition $(\alpha \mapsto \beta)$ with $Q(\alpha, \beta) > 0$, 
  \begin{align*}
	\sum_{e\in E} \mu_e(0)\mu_e(1)\E[(X,Y)\sim \+C_e]{\ell(\gamma^{X,Y})\cdot \*1[(\alpha \mapsto \beta) \in \gamma^{X,Y}]} 
    &\leq O(\ol{\lambda}^{4}\Delta^2 m) \cdot \mu(\alpha)Q(\alpha, \beta).    
  \end{align*}
\end{enumerate}
\end{lemma}

Although the definitions of congestion involve the local flipping couplings $\+C_e$ for all edges $e \in E$, 
it suffices to analyze the coupling $\+C_f$ for a single $f\in E$ using the following \emph{decoupling lemma}.

\begin{lemma}[decoupling lemma]\label{lem:pre-decoupling}
The local flipping coupling $(\+C_e)_{e\in E}$ defined in \Cref{def:new-coupling-2} satisfies the following property.
  Let $\phi:\Omega(\mu) \times \Omega(\mu)\to \mathbb{R}$ be a function. For any edge $f \in E$, it holds
  \begin{align}\label{lem:new-pre-decoupling-equ1}
    \sum_{e\in E}\mu_e(0) \mu_e(1) &\E[(X,Y)\sim \+C_e]{\phi(X,Y) \cdot \*1[X_f \neq Y_f]}\notag\\
    &\le
    \mu_f(0)\mu_f(1) \E[(X,Y)\sim \+C_f]{\tp{\phi(X,Y) +\phi(Y,X)}\cdot |{X \oplus Y}|}.
  \end{align}
\end{lemma}


\Cref{lem:pre-decoupling} is formally proved in \Cref{sec:proof-decoupling}.
We now prove \Cref{lem:expected-congestion} by assuming~\Cref{{lem:pre-decoupling}}.
Fix a transition $(\alpha \mapsto \beta)$ of $Q$. We analyze the congestion on $(\alpha \mapsto \beta)$.
For any pair $x,y\in\Omega(\mu)$, recall that $\gamma^{x,y}$ denotes the canonical path from $x$ to $y$. 
Define $\varphi,\psi:\Omega(\mu) \times \Omega(\mu)\to\mathbb{R}$ as: 
\begin{align*}
\varphi(x,y) &= \begin{cases}
    \frac{\*1[(\alpha \mapsto \beta) \in \gamma^{x,y}]}{\ell(\gamma^{x,y})} &\text{if } x,y \text{ differ at a single path or an even cycle;}\\
    0 &\text{otherwise.}
  \end{cases}\\
\psi(x,y) &= \begin{cases}
  |{X\oplus Y}|\cdot \*1[(\alpha\mapsto \beta) \in \gamma^{X,Y}] &\text{if } x,y \text{ differ at a single path or an even cycle;}\\
    0 &\text{otherwise.} 
    \end{cases}
\end{align*}

Next, fix an edge $h \in E$ such that $\alpha_h \neq \beta_h$. 
For any edge $e \in E$ and pair of configurations $(x,y) \in \Omega(\+C_e)$, note that if the transition $(\alpha \mapsto \beta) \in \gamma^{x,y}$, then $x_h \neq y_h$, because the canonical path never flips an edge where $x_h = y_h$. The expected congestion on $(\alpha \mapsto \beta)$ can be bounded as:
\begin{align} 
  \sum_{e \in E} \mu_e(0) \mu_e(1) \E[(X,Y) \sim \+C_e]{\frac{\*1[(\alpha \mapsto \beta) \in \gamma^{X,Y}]}{\ell(\gamma^{X,Y})}} 
  =\,\, &\sum_{e \in E} \mu_e(0) \mu_e(1) \E[(X,Y) \sim \+C_e]{\varphi(X,Y) \cdot \*1[X_h \neq Y_h]}\notag\\
(\text{by \Cref{lem:pre-decoupling}})\qquad    \le\,\, &\mu_h(0)\mu_h(1) \E[(X,Y)\sim \+C_h]{\tp{\varphi(X,Y) + \varphi(Y,X)}|{X \oplus Y}|} \notag\\
(\text{by }|{X \oplus Y}| \le 2 \ell(\gamma^{X,Y}))\qquad   \le\,\, &2\mu_h(0) \mu_h(1) \Pr[(X,Y) \sim \+C_h]{(\alpha \mapsto \beta) \in \gamma^{\set{X,Y}}}, \label{eq:cong-decoupling}
\end{align}
where $\gamma^{\set{X,Y}}\triangleq \gamma^{X,Y} \cup \gamma^{Y,X}$ denotes the set of transitions in either $\gamma^{X,Y}$ or $\gamma^{Y,X}$.

\filbreak

Similarly, for any fixed edge $h \in E$  such that $\alpha_h \neq \beta_h$,  the strong expected congestion on $(\alpha \mapsto \beta)$ can be bounded as:
  \begin{align}
     &\sum_{e\in E}\mu_e(0)\mu_e(1)\E[(X,Y)\sim \+C_e]{\ell(\gamma^{X,Y})\cdot \*1[(\alpha \mapsto \beta) \in \gamma^{X,Y}]} \notag\\
   \text{(by $\ell(\gamma^{X,Y})\le |{X\oplus Y}|$)}\qquad \le\,\, &\sum_{e\in E} \mu_e(0)\mu_e(1)\E[(X,Y)\sim \+C_e]{|{X\oplus Y}|\cdot  \*1[(\alpha \mapsto \beta) \in \gamma^{X,Y}]}\notag\\
   =\,\, &\sum_{e\in E}\mu_e(0) \mu_e(1) \E[(X,Y)\sim \+C_e]{\psi(X,Y) \cdot \*1[X_h \neq Y_h]}\notag\\
    \text{(by~\Cref{lem:pre-decoupling})}\qquad \le\,\, &\mu_h(0)\mu_h(1) \E[(X,Y)\sim \+C_h]{\tp{\psi(X,Y) + \psi(Y,X)}|{X \oplus Y}|} \notag\\
    =\,\, &\mu_h(0)\mu_h(1)\E[(X,Y)\sim \+C_h]{|{X\oplus Y}|^2\cdot   \*1[(\alpha \mapsto \beta) \in \gamma^{\set{X,Y}}]}. \label{eq:cong-decoupling-strong}
  \end{align}

The bounds on expected congestion and strong expected congestion in \eqref{eq:cong-decoupling} and \eqref{eq:cong-decoupling-strong} motivate us to establish the following lemma.
\begin{lemma} \label{lem:powered-cong}
    Fix any constant $p \ge 0$. For any transition $(\alpha \mapsto \beta)$ and any edge $h\in E$ such that $\alpha_h \neq \beta_h$, 
    \begin{align}\label{eq:powered-cong}
        \mu_h(0) \mu_h(1) \E[(X,Y) \sim \+C_h]{|{X \oplus Y}|^p \cdot \*1[(\alpha \mapsto \beta) \in \gamma^{\set{X,Y}}]} \le O_p(\ol{\lambda}^{p+2} \Delta^p m) \cdot \mu(\alpha) Q(\alpha,\beta),
    \end{align}
    where $\gamma^{\set{X,Y}}\triangleq \gamma^{X,Y}\cup \gamma^{Y,X}$ denotes the set of transitions appearing in either $\gamma^{X,Y}$ or $\gamma^{Y,X}$. 
\end{lemma}

With this general congestion bound, the expected congestion bound in \Cref{lem:expected-congestion} follows from \eqref{eq:cong-decoupling} and \Cref{lem:powered-cong} with $p=0$, while the strong expected congestion bound in \Cref{lem:expected-congestion} follows from \eqref{eq:cong-decoupling-strong} and \Cref{lem:powered-cong} with $p=2$.

\subsubsection{Measure-preserving flow encoding}
To complete the proof, it remains to prove \Cref{lem:pre-decoupling,lem:powered-cong}. 
We begin with \Cref{lem:powered-cong}, where we derive the general congestion bound \eqref{eq:powered-cong} by constructing a measure-preserving flow encoding.

\begin{proof}[Proof of~\Cref{lem:powered-cong}]
Fix a transition $(\alpha \mapsto \beta)$
of the Jerrum-Sinclair chain $Q$. By definition, $\alpha$ and $\beta$ differ on exactly one or two edges. We analyze each case separately.

\paragraph{Case 1: single disagreement.}
When $\abs{\alpha \oplus \beta}=1$, there is a unique edge $h$ where $\alpha_h \neq \beta_h$. 
We consider the up-transition case where $\alpha_h=0$ and $\beta_h=1$, while the down-transition case follows by symmetry. 
By the construction of the canonical path $\gamma^{x,y}$, the transition $\alpha \mapsto \beta$ can only occur in the final step of $\gamma^{x,y}$, meaning $y = \beta$. 
Moreover, since each discrepancy edge $f\in x\oplus y$ is flipped exactly once along $\gamma^{x,y}$, the transition $(\alpha \mapsto \beta) \in \gamma^{x,y}$ implies $x_h = 0$ and $y_h = 1$. 
Consequently, 
\begin{align}\label{eq:single-edge-1}
  \nonumber
  \E[(X,Y) \sim \+C_h]{\abs{X \oplus Y}^p \cdot \*1[(\alpha \mapsto \beta) \in \gamma^{\set{X,Y}}]}
  &= \E[(X,Y) \sim \+C_h]{\abs{X \oplus Y}^p \cdot \*1[(\alpha \mapsto \beta) \in \gamma^{X,Y}]}\\
  &\leq \E[(X,Y) \sim \+C_h]{\abs{X \oplus Y}^p \cdot \*1[Y = \beta]}. 
\end{align}
Applying~\Cref{lem:one-sided-CI}, there exists a constant $c_p >0$ depending on $p\ge 0$ such that the $p$-th moment of the one-sided discrepancies satisfies
\begin{align}\label{eq:one-sided-discrepancy}
    \E[(X,Y) \sim \+C_h]{\abs{X \oplus Y}^p \mid Y = \beta} \le c_p \ol{\lambda}^p \Delta^p.
\end{align}
Combining~\eqref{eq:single-edge-1} and~\eqref{eq:one-sided-discrepancy}, we obtain
\begin{align} \label{eq:bound-for-expectation-aux-1}
  \E[(X,Y) \sim \+C_h]{\abs{X \oplus Y}^p \cdot \*1[(\alpha \mapsto \beta) \in \gamma^{\set{X,Y}}]}
  &\leq c_p \ol{\lambda}^p \Delta^p \Pr[(X,Y) \sim \+C_h]{Y = \beta}
  = c_p \ol{\lambda}^p \Delta^p \mu^{h\gets 1}(y).
\end{align}
To conclude this case, we note that
\begin{align*}
  \mu_h(0)\mu_h(1) \cdot{c_p \ol{\lambda}^p \Delta^p \mu^{h\gets 1}(y)}
    =c_p \ol{\lambda}^p \Delta^p \cdot \mu_h(0) \mu(\beta)
    \le c_p \ol{\lambda}^{p+1} \Delta^p m \cdot \mu(\alpha) Q(\alpha,\beta),
\end{align*}
where the last inequality follows from the facts that $Q(\alpha,\beta)=\frac{1}{m}\min\set{1,\frac{\mu(\beta)}{\mu(\alpha)}}$ and $\mu(\beta)\le \ol{\lambda}\cdot\mu(\alpha)$.

\paragraph{Case 2: two disagreements.}
When $\abs{\alpha \oplus \beta} = 2$,  i.e., $\alpha \mapsto \beta$ is an exchange transition, there exist two edges $f,g \in E$ such that $\alpha_f = 0, \beta_f = 1$ and $ \alpha_g=1, \beta_g = 0$. 
Let $w \in f\cap g$ be the vertex shared by both $f$ and $g$. We consider the case where $h = f$, while the case $h = g$ follows by symmetry.

Given $h =f$, the event $(\alpha \mapsto \beta) \in \gamma^{\set{X,Y}}$ is equivalent to $(\alpha \mapsto \beta) \in \gamma^{X,Y}$, since $X_h = 0$ and $Y_h = 1$.
However, unlike the single-disagreement case, this event no longer implies a trivial case (e.g., $Y = \beta$ previously).
Consequently, it is not immediate why the term $\mu(\alpha) Q(\alpha, \beta)$ should appear in the right-hand side of \eqref{eq:powered-cong}. 
To proceed, we refine our understanding of  $(\alpha \mapsto \beta) \in \gamma^{X,Y}$.
%

Recall that $B=x\oplus y$ denotes the set of edges where $x$ and $y$ differ. 
We partition $B$ into $B^h\uplus B^g$ as follows:

\begin{itemize}
  \item If $B$ is a path, $w$ (the vertex shared by $h=f$ and $g$) divides $B$ into two subpaths $B^h$ and $B^g$, where $B^h$ (resp.~$B^g$) contains $h$ (resp.~$g$).
  \item If $B$ is an even-length cycle, let $w^*$ be the largest vertex in $B$. By the construction of the canonical path $\gamma^{x,y}$, we must have $w^* \neq w$ whenever $(\alpha \mapsto \beta) \in \gamma^{x,y}$. The cycle is then divided into two subpaths $B^h$ and $B^g$ by $w$ and $w^*$, with $B^h$ (resp.~$B^g$) containing $h$ (resp.~$g$).
\end{itemize}
From \Cref{def:monomer-dimer-trans-path}, we observe that during the transition $(\alpha\mapsto\beta)$ within the canonical path $\gamma^{x, y}$, the edges in $B^{h}$ are flipped, while those in $B^{g}$ (except for $g$) remain unchanged.
Formally, this is described by the following condition:
  \begin{align}\label{eq:aorb}
    \alpha \cup \beta = x \triangle {B^{h}} = y \triangle {B^{g}},
  \end{align}
  where we slightly abuse the notation by interpreting the configurations $\alpha,\beta,x,y\in\{0,1\}^E$ as subsets of edges $\alpha,\beta,x,y\subseteq E$,
  with the corresponding configurations as indicator vectors.
  Therefore, the union $\alpha\cup\beta$ produces the same set of edges as the symmetric differences $x\triangle {B^{h}} = y \triangle {B^{g}}$.
  
  Note that $\alpha \cup \beta$ is no longer a matching in $G$, as both $h$ and $g$ are in $\alpha \cup \beta$. 
  To resolve this issue, we introduces a new graph $\wh{G}=(\wh{V},\wh{E})$.
  The graph $\wh{G}=(\wh{V},\wh{E})$ is constructed from $G$ by first removing every edge incident to $w$ except for $g$ and $h$, and then merging the two edges $g$ and $h$ into a single edge $h^{\star}$. This induces a natural mapping $\textsf{proj}$ that maps a configuration $x \in \{0,1\}^E$ with $x_h = x_g$ to a new configuration $\hat{x} \in \{0,1\}^{\wh{E}}$ as follows:
	\begin{align*}
	\forall e \in \hat{E},\quad	\textsf{proj}(x)_e \triangleq  \begin{cases}
			x_e,& \text{if $e\ne h^{\star}$}, \\
			x_h=x_g,& \text{if $e=h^{\star}$}.
		\end{cases}
	\end{align*}
    It is obvious that if $x_h = x_g=1$ and $g,h$ are the only edges that violate the constraint of $x$ being a matching, then $\textsf{proj}(x)$ is a matching in $\wh{G}$. 
	Let $\wh{\mu}$ be the monomer-dimer distribution on $\wh{G}$ with the same edge weight $\lambda$. We use the analogous natation $\wh{\+{C}}_{h^{\star}}$ to denote the local flipping coupling (as defined in \Cref{def:new-coupling-2}) with respect to the monomer-dimer distribution $\wh{\mu}$ at the edge $h^\star$.

	For $(x,y)\in\Omega(\+{C}_h)$ such that $(\alpha \mapsto \beta) \in \gamma^{x,y}$, if the disagreement set  $B=x\oplus y$ forms a cycle, let $e$ be the unique edge in $B^{h}$ that is incident to $w^*$. We define a \emph{measure-preserved flow encoding} that maps from $\Omega(\+{C}_h)$ to $\Omega({\wh{\+{C}}_{h^{\star}}})$ as follows: 
\begin{align} \label{eq:def-enc}
  \textsf{enc}(x,y) \triangleq  \begin{cases}
  \tp{\textsf{proj}(x\triangle{B^{g}}), \textsf{proj}(x\triangle{B^{h}})},& \text{if $B$ is a path;} \\
  \tp{\textsf{proj}(x\triangle{B^{g}}\triangle{\{e\}}), \textsf{proj}(x\triangle{B^{h}})},& \text{if $B$ is an even-length cycle;} 
  \end{cases}
\end{align}
which is equivalent to the following by~\eqref{eq:aorb}:
\begin{align*}
  \textsf{enc}(x,y) \triangleq  \begin{cases}
  \tp{\textsf{proj}(y\triangle{B^{h}}), \textsf{proj}(y\triangle{B^{g}})},& \text{if $B$ is a path;} \\
  \tp{\textsf{proj}(y\triangle{B^{h}}\triangle{\{e\}}), \textsf{proj}(y\triangle{B^{g}})},& \text{if $B$ is an even-length cycle.} 
  \end{cases}
\end{align*}

The encoding defined here is inspired by the injective mapping introduced in the original canonical path method by Jerrum and Sinclair~\cite{jerrum1989approximating}.
Specifically, our encoding function maps every $(x,y)\in\Omega(\+{C}_h)$ with $(\alpha\mapsto\beta)\in\gamma^{x,y}$ injectively to an $(\wh{x},\wh{y})\in\Omega({\wh{\+{C}}_{h^{\star}}})$ with $\wh{y}=\textsf{proj}(\alpha \cup \beta)$, and approximately preserves the relevant measure between the two couplings $\+{C}_h$ and ${\wh{\+{C}}_{h^{\star}}}$, as its name suggests. 
This is formalized by the following lemma.
\begin{lemma}\label{lem:enc}
  The function $\textsf{enc}$ defined in \eqref{eq:def-enc} is an injection from $\{(x,y)\in\Omega(\+{C}_h)\mid(\alpha\mapsto\beta)\in\gamma^{x,y}\}$ to $\{(\wh{x},\wh{y})\in\Omega(\wh{\+{C}}_{h^{\star}})\mid\wh{y}=\textsf{proj}(\alpha\cup\beta)\}$. 
  Furthermore, for any $(x,y)\in\Omega(\+{C}_h)$ such that $(\alpha\mapsto\beta)\in\gamma^{x,y}$, 
  \begin{align*}
    \mu_h(0)\mu_h(1)\cdot\Pr[(X,Y)\sim\+{C}_h]{(X,Y) = (x,y)} \le \eta\cdot\wh{\mu}_{h^{\star}}(0)\wh{\mu}_{h^{\star}}(1)\cdot \Pr[(\wh{X},\wh{Y}) \sim\wh{\+C}_{h^{\star}}]{(\wh{X},\wh{Y}) = \mathsf{enc}(x,y)},
  \end{align*}
  where $\eta\triangleq \tp{\wh{Z}/Z}^2\cdot\ol{\lambda}^2$, and $Z$ (resp.~$\wh{Z}$) denotes the partition function of $\mu$ (resp.~$\wh{\mu}$).
\end{lemma}

As a direct corollary of \Cref{lem:enc}, we can rewrite the expectation on the left-hand side of \eqref{eq:powered-cong} in the new probability space of $\wh{\+C}_{h^\star}$, replacing the indicator for the non-trivial event $(\alpha\mapsto \beta) \in \gamma^{\set{X,Y}}$ with the indicator for the simpler event  $\wh{Y} = \textsf{proj}(\alpha \cup \beta)$.
We formalize this as follows:
\begin{corollary} \label{cor:enc}
  Let $\eta$ be the factor defined in \Cref{lem:enc}. Then, it holds that
  \begin{align*} 
    \mu_h(0)\mu_h(1)&\cdot\E[(X,Y)\sim {\+C}_h]{\abs{X\oplus Y}^{p} \cdot \*{1}[(\alpha\mapsto\beta)\in\gamma^{\set{X,Y}}]} \\
    &\leq \eta \cdot \wh{\mu}_{h^{\star}}(0)\wh{\mu}_{h^{\star}}(1) \E[(\wh{X},\wh{Y})\sim\wh{\+{C}}_{h^{\star}}]{\tp{\abs{\wh{X}\oplus\wh{Y}} + 2}^{p} \cdot \*{1}[\wh{Y}=\mathsf{proj}(\alpha\cup\beta)]}.
  \end{align*}
\end{corollary}
The proof of \Cref{cor:enc} is a straightforward calculation, which we will provide later. Now, assuming \Cref{cor:enc}, we can complete the proof of \Cref{lem:powered-cong} using the same argument as before.
By applying the same reasoning  that led to \eqref{eq:bound-for-expectation-aux-1},  there exists a constant $C_p > 0$ such that
\begin{align*}
  \E[(\wh{X},\wh{Y})\sim\wh{\+{C}}_{h^{\star}}]{\tp{\abs{\wh{X}\oplus\wh{Y}} + 2}^{p} \cdot \*{1}[\wh{Y}=\mathsf{proj}(\alpha\cup\beta)]} \leq C_p \ol{\lambda}^p\Delta^p \wh{\mu}^{h^\star\gets 1}(\mathsf{proj}(\alpha \cup \beta)).
\end{align*}
Recall that $\eta = \tp{\wh{Z}/Z}^2\cdot\ol{\lambda}^2$ and ${\wh{Z}}\le {Z}$.
We can conclude the proof by noting:
\begin{align*}
  \eta \cdot \wh{\mu}_{h^{\star}}(0)\wh{\mu}_{h^{\star}}(1) \wh{\mu}^{h^\star\gets 1}(\textsf{proj}(\alpha \cup \beta))
  \leq  \ol{\lambda}^2 \cdot \frac{\hat{Z}}{Z} \cdot\wh{\mu}(\textsf{proj}(\alpha \cup \beta))
  = \ol{\lambda}^2  \cdot\mu(\alpha)
  =\ol{\lambda}^2 m \cdot \mu(\alpha) Q(\alpha, \beta),
\end{align*}
where the first equality follows from the fact that $\abs{\-{proj}(\alpha \cup \beta)} = \abs{\alpha}=\abs{\beta}$; and in the last equality, we use the fact that $m \cdot \mu(\alpha)Q(\alpha, \beta) = \min\set{\mu(\alpha),\mu(\beta)} = \mu(\alpha)$, since $\abs{\alpha} = \abs{\beta}$.
\end{proof}

\begin{proof}[Proof of \Cref{cor:enc} assuming~\Cref{lem:enc}]
Assuming \Cref{lem:enc}, we now prove \Cref{cor:enc} using the bound on the $p$-th moment of the one-sided discrepancy (\Cref{lem:one-sided-CI}). Since $(\alpha \mapsto \beta) \in \gamma^{x,y}$ implies that $x_h = 0$ and $y_h = 1$  (as each disagreeing edge between $x$ and $y$ is flipped exactly once), it follows that
\begin{align*}
  \mu_h(0)\mu_h(1)&\cdot\E[(X,Y)\sim {\+C}_h]{\abs{X\oplus Y}^{p} \cdot \*{1}[(\alpha\mapsto\beta)\in\gamma^{\set{X,Y}}]} \\
    &= \mu_h(0)\mu_h(1)\cdot\E[(X,Y)\sim {\+C}_h]{\abs{X\oplus Y}^{p} \cdot \*{1}[(\alpha\mapsto\beta)\in\gamma^{X,Y}]}.
\end{align*}
By \Cref{lem:enc}, we have
\begin{align} \label{eq:congestion-1}
    \nonumber
    \mu_h(0) \mu_h(1) & \E[(X,Y)\sim {\+C}_h]{\abs{X\oplus Y}^{p} \cdot \*{1}[(\alpha\mapsto\beta)\in\gamma^{X,Y}]} \\
    \nonumber &= \sum_{(x,y)\in\Omega(\+{C}_h) \atop (\alpha\mapsto\beta)\in\gamma^{x,y}} \abs{x\oplus y}^{p}\cdot \mu_h(0) \mu_h(1) \cdot\Pr[(X,Y) \sim \+{C}_h]{(X,Y) = (x,y)}\\
    &\le \eta\cdot\wh{\mu}_{h^{\star}}(0)\wh{\mu}_{h^{\star}}(1)\cdot \sum_{(x,y)\in\Omega(\+{C}_h) \atop (\alpha\mapsto\beta)\in\gamma^{x,y}} \abs{x\oplus y}^{p}\cdot \Pr[(\wh{X},\wh{Y}) \sim \wh{\+C}_{h^{\star}}]{(\wh{X}, \wh{Y}) = \textsf{enc}(x,y)}.
\end{align}
Since $\textsf{enc}$ is an injection (by \Cref{lem:enc}),
 the summation over all configuration pairs $(x,y) \in \Omega(\+C_h)$ with $(\alpha \mapsto \beta) \in \gamma^{x,y}$ can be bounded by summing over all $(\hat{x},\hat{y}) \in \Omega(\hat{\+C}_{h^\star})$ with $\hat{y} = \textsf{proj}(\alpha\cup\beta)$. Specifically, 
\begin{align}\label{eq:congestion-2}
  \nonumber
  \sum_{(x,y)\in\Omega(\+{C}_h) \atop (\alpha\mapsto\beta)\in\gamma^{x,y}} & \abs{x\oplus y}^{p}\cdot\Pr[(\wh{X},\wh{Y})\sim \wh{\+C}_{h^{\star}}]{(\wh{X}, \wh{Y}) = \textsf{enc}(x,y)} \\
  &\le \sum_{(\wh{x},\wh{y})\in\Omega(\wh{\+{C}}_{h^{\star}}) \atop \wh{y}=\textsf{proj}(\alpha\cup\beta)}\tp{\abs{\wh{x}\oplus\wh{y}} + 2}^{p}\cdot\Pr[(\wh{X},\wh{Y}) \sim \wh{\+C}_{h^{\star}}]{(\wh{X}, \wh{Y}) = (\wh{x},\wh{y})}.
\end{align}
\Cref{cor:enc} follows immediately from~\eqref{eq:congestion-1} and~\eqref{eq:congestion-2}.
\end{proof}

Next, we prove \Cref{lem:enc}, ensuring the properties of the measure-preserving flow encoding.
\begin{proof}[Proof of~\Cref{lem:enc}]
Fix any transition $\alpha \mapsto \beta$.
For each $(x,y)\in\Omega(\+{C}_h)$ such that $(\alpha\mapsto\beta)\in\gamma^{x,y}$, we let $(\wh{x},\wh{y})=\textsf{enc}(x,y)$. 
We prove the following three properties sequentially.

\paragraph{Property 1: Codomain.}
  By~\eqref{eq:aorb}, we have $x\triangle{B^{h}}=\alpha\cup\beta$. From the definition of $\textsf{enc}$ in~\eqref{eq:def-enc}, $\wh{y}=\textsf{proj}(x\triangle{B^{h}})=\textsf{proj}(\alpha\cup\beta)$. Moreover, $\wh{x}_{h^\star} = \*1[g\in x \triangle B^g] = 0$ since $x_g = \alpha_g = 1$.
  To verify that  $(\wh{x},\wh{y}) \in \Omega(\wh{\+C}_{h^\star})$, it remains to ensure $\wh{B} \triangleq \wh{x} \oplus \wh{y}$ forms a path or an even-length cycle containing $h^\star$. 
  If $B=x\oplus y$ is a path, then $\hat{B}=B\cup\set{h^\star}\setminus \set{g,h}$, and if $B$ is a cycle, then $\hat{B}=B\cup\set{h^\star}\setminus \set{g,h,e}$, with $e,g,h$ as given in~\eqref{eq:def-enc}.  In both cases, since $h^\star$ is the merge of $g$ and $h$, we have $\wh{B}$ forms a path. 
	
\paragraph{Property 2: Injection.}
To prove that $\textsf{enc}$ (with respect to $\alpha \mapsto \beta$) is an injection, it is sufficient to show that the preimage $(x, y)$ can be uniquely recovered from the image $(\wh{x}, \wh{y})$ given $\alpha \mapsto \beta$.
Note that if the disagreement set $B=x\oplus y$ can be recovered, then $x=(\alpha\cup\beta)\triangle B^h$ and $y=(\alpha\cup\beta)\triangle B^g$ are determined based on the identity in~\eqref{eq:aorb}. 
  It follows from the previous property that $\wh{B}$ always forms a path. We denote the endpoints of path $\wh{B}$ by $s$ and $t$. 
  Then, the construction of $B$ has only two choices: either $B$ is a path satisfying $B = \wh{B} \cup \{g,h\} \setminus \{h^\star\}$, or $B$ is a cycle satisfying $B = \wh{B} \cup \{g,h, \{s,t\}\} \setminus \{h^\star\}$, where $\{s,t\}$ is an edge in $E$. 
  
  To distinguish between these cases, let $h^\star = \{w_g,w_h\}$, where $w_g$ and $w_h$ are the vertices incident to edges $g$ and $h$, respectively. 
  The edge $h^\star$ divides $\wh{B}$ into two subpaths: $\wh{B}^g$, from $s$ to $w_g$ (with $w_g=s$ if $\wh{B}^g=\emptyset$) and $\wh{B}^h$, from $w_h$ to $t$ (with $w_h=t$ if $\wh{B}^h=\emptyset$).
  We claim the distinction between the path case and cycle case of $B$ follows from the order between $s$ and $t$  (in the total order $(V,<)$ used in~\Cref{def:monomer-dimer-trans-path}):
    \begin{enumerate}
      \item If $s<t$, then $B$ forms a path, and thus $B=\wh{B}\cup \set{g,h}\setminus \set{h^\star}$ ;
      \item If $s>t$, then $B$  forms a cycle, and thus $=\wh{B}\cup\set{g,h,\set{s,t}}\setminus \set{h^\star}$. 
    \end{enumerate}
   To see this is correct, note that~\Cref{def:monomer-dimer-trans-path} specifies a ``flipping order'' for edges in $B= x \oplus y$. 
       When $B$ is a path with endpoints $u > v$, the canonical path flips the edges along the path from $u$ to $v$. Here, $t$ is an endpoint of $\hat{B}^h$, implying it is also an endpoint in $B^h$, which forces $s = v < u = t$.
       When $B$ is a cycle, the flipping order ensures that the removed edge $e=\{w^*,w\}$ incident to the largest vertex $w^*$ in $B$, where $w$ is an endpoint in $\hat{B}^h$, which forces $t = w < w^* = s$.

Therefore, $B$ (and hence $(x,y)$) can be uniquely recovered and the function $\mathsf{enc}$ is injective.

\paragraph{Property 3: Approximate measure-preserving.}
   By~\Cref{lem:explicit-formula}, the joint probabilities satisfy
  \begin{align}
    \label{eq:num}\mu_h(0)\mu_h(1)\cdot\Pr[(X,Y) \sim\+C_{h}]{(X, Y) = (x,y)} &= \mu_{{\overline{\partial}} B}(x_{{\overline{\partial}} B})\cdot\mu(y) \\
    \label{eq:denom}\wh{\mu}_{h^{\star}}(0)\wh{\mu}_{h^{\star}}(1)\cdot\Pr[(\wh{X},\wh{Y}) \sim\wh{\+C}_{h^{\star}}]{(\wh{X},\wh{Y}) = (\hat{x},\hat{y})} &= \wh{\mu}_{{\overline{\partial}}\wh{B}}(\wh{x}_{{\overline{\partial}}\wh{B}})\cdot\wh{\mu}(\wh{y})
  \end{align}
  Since $x$ and $\hat{x}$ (resp.~$y$ and $\hat{y}$) agree  over $E \setminus \ol{\partial} B = \hat{E} \setminus \ol{\partial} \hat{B}$, we have
  $$
\mu^{x_{{\overline{\partial}} B}}_{E\setminus{\overline{\partial}} B}(x_{E\setminus{\overline{\partial}} B})=\wh{\mu}^{\wh{x}_{{\overline{\partial}}\wh{B}}}_{\wh{E}\setminus{\overline{\partial}}\wh{B}}(\wh{x}_{\wh{E}\setminus{\overline{\partial}}\wh{B}}).$$
Therefore, taking the ratio of~\eqref{eq:num} and~\eqref{eq:denom}, we obtain
\begin{align*}
    \frac{\text{\eqref{eq:num}}}{\text{\eqref{eq:denom}}} = \frac{\mu(x) \mu(y)}{\hat{\mu}(\hat{x}) \hat{\mu} (\hat{y})} = \tp{\frac{\hat{Z}}{Z}}^2\cdot\frac{\lambda^{\norm{x}_1}\cdot\lambda^{\norm{y}_1}}{\lambda^{\norm{\wh{x}}_1}\cdot\lambda^{\norm{\wh{y}}_1}} &=\tp{\frac{\wh{Z}}{Z}}^2\cdot\frac{\lambda^{\norm{x_{{\overline{\partial}} B}}_1+\norm{y_{{\overline{\partial}} B}}_1}}{\lambda^{\norm{\wh{x}_{{\overline{\partial}} \wh{B}}}_1+\norm{\wh{y}_{{\overline{\partial}} \wh{B}}}_1}}\\
    &= \tp{\frac{\wh{Z}}{Z}}^2\cdot \lambda^{\abs{B} - \abs{\hat{B}}}
    \le \tp{\frac{\wh{Z}}{Z}}^2 \cdot \ol{\lambda}^2.
\end{align*}
Here, the 1-norm $\norm{\cdot}_1$ counts the number of 1's in a configuration, and the final inequality follows from the construction of $B=x\oplus y$ and $\wh{B}=\wh{x}\oplus\wh{y}$, where $(\wh{x},\wh{y})=\mathsf{enc}(x,y)$. 
\end{proof}

\subsubsection{Decoupling the local flipping coupling}\label{sec:proof-decoupling}
We now prove \Cref{lem:pre-decoupling}, the decoupling lemma for the local flipping coupling.
\begin{proof}[Proof of~\Cref{lem:pre-decoupling}]
  By \Cref{lem:explicit-formula}, for every edge $e \in E$, the support $\Omega(\+{C}_e)$ of the local flipping coupling $\+{C}_e$ satisfies
  \begin{align*}
    \Omega(\+{C}_e)=\set{(\sigma,\tau)|\text{$\sigma\oplus\tau$ forms a path or an even-length cycle, and }\sigma(e)=0,\tau(e)=1}.
  \end{align*}
  This implies that for any $e,f \in E$ and any $(\sigma,\tau) \in \Omega(\+{C}_e)$, if $\sigma_{f} \ne \tau_{f}$, then either $(\sigma,\tau) \in \Omega(\+C_f)$ or $(\tau,\sigma) \in \Omega(\+C_f)$.
  Thus, we can split the right-hand side of \eqref{lem:new-pre-decoupling-equ1}  as follows:
  \begin{align*}
    \sum_{e\in E}\mu_e(0)\mu_e(1) & \E[(X,Y)\sim \+C_e]{\phi(X,Y)\cdot \*1[X_f \ne Y_f]} \\
    &\leq \sum_{e \in E}\mu_e(0)\mu_e(1)\E[(X,Y) \sim \+C_e]{\phi(X,Y)\tp{\*1[(X,Y) \in \Omega(\+C_f)] + \*1[(Y,X) \in \Omega(\+C_f)]}}.
  \end{align*}
To complete the proof, it suffices to show that
  \begin{align*}
    \sum_{e \in E}\mu_e(0)\mu_e(1)\E[(X,Y) \sim \+C_e]{\phi(X,Y)\cdot \*1[(X,Y) \in \Omega(\+C_f)]}
    \le \mu_f(0)\mu_f(1) \E[(X,Y)\sim \+C_f]{\phi(X,Y)\cdot \abs{X \oplus Y}}.
  \end{align*}
  The symmetric case with $(Y,X) \in \Omega(\+C_f)$ follows by symmetry.
  
  By \Cref{lem:explicit-formula}, there is a detailed identity between $\+C_e$ and $\+C_f$  for any $e,f \in E$. 
  Specifically, for every $(\sigma, \tau) \in \Omega(\+C_e) \cap \Omega(\+C_f)$, it holds
  \begin{align} \label{eq:relation-Ce-Cf}
    \mu_e(0)\mu_e(1)\Pr[(X,Y)\sim\+C_e]{(X,Y) = (\sigma,\tau)} &= \mu_f(0)\mu_f(1) \Pr[(X,Y)\sim\+C_f]{(X,Y)=(\sigma,\tau)}
  \end{align}
  Applying this identity, we have
  \begin{align*}
    \sum_{e \in E}&\mu_e(0)\mu_e(1) \E[(X,Y) \sim \+C_e]{\phi(X,Y)\cdot \*1[(X,Y) \in \Omega(\+C_f)]} \\
    &= \sum_{e\in E} \sum_{(\sigma,\tau) \in \Omega(\+{C}_e)\cap \Omega(\+{C}_f)} \phi(\sigma,\tau)\cdot  \mu_e(0)\mu_e(1) \Pr[(X,Y)\sim \+C_e]{(X,Y) = (\sigma,\tau)} \\
    &= \sum_{e\in E} \sum_{(\sigma,\tau) \in \Omega(\+C_e) \cap \Omega(\+C_f) } \phi(\sigma,\tau)\cdot  \mu_f(0)\mu_f(1) \Pr[(X,Y)\sim \+C_f]{(X,Y) = (\sigma,\tau)} \tag{by \eqref{eq:relation-Ce-Cf}} \\
    &= \mu_f(0)\mu_f(1)  \sum_{(\sigma,\tau) \in \Omega(\+C_f) }\sum_{e\in E}\*1[(\sigma,\tau)\in \Omega(\+C_e)] \cdot \phi(\sigma,\tau)\Pr[(X,Y)\sim \+C_f]{(X,Y) = (\sigma,\tau)} \\
    &\leq \mu_f(0)\mu_f(1) \sum_{(\sigma,\tau) \in \Omega(\+C_f)} \phi(\sigma,\tau)\cdot \abs{\sigma \oplus \tau} \Pr[(X,Y)\sim \+C_f]{(X,Y) = (\sigma,\tau)} \\
    &= \mu_f(0)\mu_f(1) \E[(X,Y)\sim \+C_f]{\phi(X,Y)\cdot \abs{X \oplus Y}}.
  \end{align*}
  This completes the proof.
\end{proof}

\subsection{Congestion bounds under pinnings}\label{sec:length-and-congestion-analysis-for-conditional-distributions}

We have established discrepancy and congestion bounds without pinning, as stated in \Cref{thm:effective-trans-flow}. 
To extend these bounds to the entire family $\mathfrak{Q}_{\JS}$ of Jerrum-Sinclair chains, we must address the general case that includes pinnings.
%
Fix a subset of edges $\Lambda \subseteq E$ and a feasible pinning $\tau \in \{0,1\}^{E \setminus \Lambda}$. 
Consider the conditional distribution $\mu^\tau$ and the Jerrum-Sinclair chain $Q^\tau$ for $\mu^\tau$ as defined in \Cref{sec:app-JS-chain}.
Notably, this differs slightly from the Jerrum-Sinclair chain obtained via self-reduction under the pinning $\tau$.
To see this, define the subset $S \subseteq \Lambda$ of ``free'' edges under $\tau$ as
\begin{align*}
  S \triangleq \set{f \in \Lambda \mid \Omega(\mu^\tau_f) = \{0,1\}}.
\end{align*}
Let $\hat{G}$ be the subgraph of $G$ induced by the edge set $S$, and let $\hat{\mu}$ denote the monomer-dimer distribution on $\hat{G}$ with edge weight $\lambda$. It is easy to see that $\hat{\mu} = \mu^\tau_S$.
Now, let $\hat{Q}$ be the Jerrume-Sinclair chain for the monomer-distribution $\hat{\mu}$, which corresponds to the chain obtained via self-reduction.
Since \Cref{thm:effective-trans-flow} directly applies to $\hat{Q}$, our goal is to extend these bounds to $Q^\tau$.


The only difference between $Q^\tau$ and $\hat{Q}$ is that $Q^\tau$ is a lazy version of $\hat{Q}$: when an edge $e\in\Lambda\setminus S$  is selected for an update in $Q^\tau$,  no change occurs, effectively introducing a lazy probability $1-\frac{|S|}{|\Lambda|}$.
For any $e \in S$, we apply \Cref{def:monomer-dimer-trans-path} on $\hat{\mu}$ to construct the transport flow $\hat{\Gamma}_e$ from $\hat{\mu}^{e \gets 0} = \mu^{\tau \land e \gets 0}_S$ to $\hat{\mu}^{e \gets 1} = \mu^{\tau \land e \gets 1}_S$. 
Since the configuration on $E \setminus S$ remains fixed in $\mu^\tau$, we can extend $\hat{\Gamma}_e$ over $E \setminus S$ (as it is fixed in $\mu^\tau$) to obtain the transport flow $\Gamma_e$ from $\mu^{\tau \land e \gets 0}$ to $\mu^{\tau \land e \gets 1}$.

\filbreak

We now establish that  $\Gamma_e$ satisfies the following properties with respect to the chain $Q^\tau$:
\begin{enumerate}
  \item\label{lem:trans-flow-length} $O(\ol{\lambda}\Delta  \log^2(1+\ol{\lambda}))$-expected squared length;
  \item\label{lem:trans-flow-congestion}\label{lem:trans-flow-strong-congestion} $O(\ol{\lambda}^2  |\Lambda|)$-expected congestion, and
   strong $O(\ol{\lambda}^4 \Delta^2 |\Lambda|)$-expected congestion.
  \end{enumerate}

Since the extension from $\hat{\Gamma}_e$ to ${\Gamma}_e$ does not alter path length, and \Cref{thm:effective-trans-flow} applies directly to $\hat{\Gamma}_e$ for $\hat{Q}$, the bound of $O(\ol{\lambda}\Delta  \log^2(1+\ol{\lambda}))$-expected squared length follows immediately. 

In contrast, $\hat{\Gamma}_e$ and ${\Gamma}_e$ may have different congestion because $Q^\tau$ is lazier than $\hat{Q}$.
To bound the congestion of $\Gamma_e$,
we apply \Cref{thm:effective-trans-flow} to $\hat{\Gamma}_e$ for $\hat{Q}$. 
For any transition $\hat{x} \mapsto \hat{y}$ in $\hat{Q}$,
\begin{align*}
  \sum_{e\in S} \hat{\mu}_e(0)\hat{\mu}_e(1) \E[\hat{\gamma} \sim \hat{\Gamma}_e]{\frac{\*1[(\hat{x} \mapsto \hat{y}) \in \hat{\gamma}]}{\ell(\hat{\gamma})}} \leq \hat{\kappa} \cdot \hat{\mu}(\hat{x})\hat{Q}(\hat{x},\hat{y}),\quad \text{where $\hat{\kappa} = O(\ol{\lambda}^2 |S|)$}.
\end{align*}
Now, consider the chain $Q^\tau$ for $\mu^\tau$. For any transition $x \mapsto y$ in $Q^\tau$ with $x \neq y$, we have:
\begin{itemize}
  \item $x_S \mapsto y_S$ is a transition in $\hat{Q}$;
  \item $\frac{|S|}{|\Lambda|}\cdot \hat{Q}(x_S,y_S) = Q^\tau(x,y)$, due to the laziness of $Q^\tau$ from selecting an $e\in\Lambda\setminus S$.
\end{itemize}
Thus, for any transition $(x \mapsto y)$ with $x \neq y$,
the expected congestion of $\Gamma_e$ satisfies
\begin{align*}
  &\sum_{e \in \Lambda:0,1 \in \Omega(\mu^\tau_e)} \mu^\tau_e(0)\mu^\tau_e(1) \E[\gamma \sim \Gamma_e]{\frac{\*1[(x \mapsto y) \in \gamma]}{\ell(\gamma)}} = \sum_{e \in S} \hat{\mu}_e(0)\hat{\mu}_e(1) \E[\gamma \sim {\Gamma}_e]{\frac{\*1[(x \mapsto y) \in {\gamma}]}{\ell({\gamma})}} \\
  =\,& \sum_{e \in S} \hat{\mu}_e(0)\hat{\mu}_e(1) \E[\hat{\gamma} \sim \hat{\Gamma}_e]{\frac{\*1[(x_S \mapsto y_S) \in \hat{\gamma}]}{\ell(\hat{\gamma})}} \leq \hat{\kappa} \cdot \hat{\mu}(x_S)\hat{Q}(x_S,y_S) = \frac{|\Lambda|}{|S|} \cdot O(\ol{\lambda}^2 |S|) \cdot \mu^\tau(x) Q^\tau(x,y)\\
  =\,& O(\ol{\lambda}^2 |\Lambda|) \cdot \mu^\tau(x) Q^\tau(x,y).
\end{align*}
This proves the bound of $O(\ol{\lambda}^2  |\Lambda|)$-expected congestion for the transport flow $\Gamma_e$ in the chain $Q^\tau$.
The bound of strong $O(\ol{\lambda}^4 \Delta^2 |\Lambda|)$-expected congestion follows by applying a similar argument.


\subsection{Concavity of the Dirichlet forms and marginal bounds}\label{sec:concave-Dirichlet-forms}
We verify that the Jerrum-Sinclair chains satisfy the concavity condition stated in \Cref{def:concave-Dirichlet-forms}.

\begin{proof}[Proof of~\Cref{lem:concave-JS}]
  Let $\tau \in \{0,1\}^{E \setminus \Lambda}$ be a feasible pinning, and let $Q^{\tau}\in\mathfrak{Q}_{\JS}$ be the Jerrum-Sinclair chain for $\mu^\tau$.
  For any $\alpha,\beta \in \Omega(\mu^\tau)$, we claim the following inequality holds:
  \begin{align}\label{eq:assump-concave}
      \frac{1}{\abs{\Lambda}}\sum_{e \in \Lambda, c \in \{0,1\}} \*1[\alpha_e = \beta_e = c] \cdot  Q^{\tau \wedge (e \gets c)}(\alpha,\beta) \le Q^{\tau}(\alpha,\beta).
  \end{align}
  Assuming~\eqref{eq:assump-concave} holds, for any $f: \Omega \to \mathbb{R}$, 
  by a straightforward calculation, we obtain:
  \begin{align*}
      &\frac{1}{\abs{\Lambda}}\sum_{e \in \Lambda} \E[c \sim \mu^\tau_e]{\+E_{Q^{\tau \wedge (e \gets c)}}(f,f)} \\
      = &\frac{1}{2\abs{\Lambda}} \sum_{e \in \Lambda, c \in \{0,1\}} \mu^\tau_e(c) \sum_{\alpha,\beta \in \Omega(\mu^{\tau \wedge (e \gets c)})}\mu^{\tau \wedge (e \gets c)}(\alpha) \cdot Q^{\tau \wedge (e \gets c)}(\alpha,\beta) \cdot (f(\alpha)-f(\beta))^2\\
      = &\frac{1}{2\abs{\Lambda}} \sum_{\substack{e \in \Lambda, c \in \{0,1\}\\\alpha,\beta \in \Omega(\mu^\tau)}} \*1[\alpha_e = \beta_e = c] \cdot \mu^\tau(\alpha) \cdot Q^{\tau \wedge (e \gets c)}(\alpha,\beta) \cdot (f(\alpha)-f(\beta))^2\\
      \le &\frac{1}{2} \sum_{\alpha,\beta \in \Omega(\mu^\tau)} \mu^\tau (\alpha) \cdot Q^\tau(\alpha,\beta) \cdot (f(\alpha)-f(\beta))^2 = \+E_{Q^\tau}(f,f).
  \end{align*}
  It remains to verify~\eqref{eq:assump-concave}. We discuss all possible types of transition $(\alpha \mapsto \beta)$:
  \begin{enumerate}
  \item $\alpha \oplus \beta=\{f\}$ for an edge $f \in \Lambda$. In this case, for any edge $e \neq f$ and spin state $c = \alpha_e = \beta_e$, the transition happens in Markov chain $Q^{\tau \wedge (e \gets c)}$ with probability 
  \begin{align*}Q^{\tau \wedge (e \gets c)}(\alpha,\beta) &= \frac{1}{\abs{\Lambda}-1} \cdot \min \set{1,\frac{\mu^{\tau \wedge (e \gets c)}(
  \beta)}{\mu^{\tau \wedge (e \gets c)}(\alpha)}}\\
  &=\frac{1}{\abs{\Lambda}-1} \cdot \min \set{1,\frac{\mu^{\tau}(
  \beta)}{\mu^{\tau}(\alpha)}} =  \frac{\abs{\Lambda}}{\abs{\Lambda}-1} Q^{\tau}(\alpha,\beta),
  \end{align*}
  as  $Q^{\tau \wedge (e \gets c)}$ chooses an edge $\widetilde{e}$ in $\Lambda \setminus \{e\}$ uniformly at random. Hence,
  \begin{align*}
      \frac{1}{\abs{\Lambda}}\sum_{e \in \Lambda, c \in \{0,1\}} \*1[\alpha_e = \beta_e = c] \cdot  Q^{\tau \wedge (e \gets c)}(\alpha,\beta) = \frac{1}{\abs{\Lambda}} \sum_{e \in \Lambda, e\neq f} Q^{\tau \wedge (e \gets \alpha_e)}(\alpha,\beta) = Q^{\tau}(\alpha,\beta).
  \end{align*}
  \item $\alpha\oplus\beta=\{f,g\}$ and $\alpha_f = 0$ for distinct edges $f,g \in \Lambda$. In this case, for any edge $e \not\in \{f,g\}$ and $c = \alpha_e = \beta_e$, the transition happens in Markov chain $Q^{\tau \land (e \gets c)}$ with probability
  \begin{align*}
      Q^{\tau \wedge (e \gets c)}(\alpha,\beta) = \frac{1}{\abs{\Lambda}-1} \cdot \min\set{1,\frac{\mu^{\tau \wedge (e \gets c)}(\alpha)}{\mu^{\tau \wedge (e \gets c)}(\beta)}} = \frac{\abs{\Lambda}}{\abs{\Lambda}-1} Q^\tau(\alpha,\beta),
  \end{align*}
  as  $Q^{\tau \wedge (e \gets c)}$ chooses an edge $\widetilde{e}$ in $\Lambda \setminus \{e\}$ uniformly at random. Hence,
\begin{align*}
      \frac{1}{\abs{\Lambda}}\sum_{e \in \Lambda, c \in \{0,1\}} \*1[\alpha_e = \beta_e = c] \cdot  Q^{\tau \wedge (e \gets c)}(\alpha,\beta) = \frac{1}{\abs{\Lambda}} \sum_{e \in \Lambda, e\not\in \{f,g\} } Q^{\tau \wedge (e \gets \alpha_e)}(\alpha,\beta) \leq Q^{\tau}(\alpha,\beta).
  \end{align*}
  \end{enumerate}
  This concludes the proof of~\eqref{eq:assump-concave}.
\end{proof}

The following proposition verifies the marginal lower bound for monomer-dimer distributions.
\begin{proposition}\label{lem:marginal-coef}
Let $\mu$ be the monomer-dimer model on a graph $G=(V,E)$ with edge weight $\lambda>0$. 
The marginal lower bound satisfies $\phi\triangleq\min \set{\mu_e(c) \mid e \in E, c\in [q], \mu_e(c) \neq 0}=\Omega\tp{\frac{\lambda}{(1+\lambda\Delta)^2}}$.
\end{proposition}
\begin{proof}
  For edge $e = \{u,v\}$, it is well known that $\mu_e(0) \ge \frac{1}{1+\lambda}$.
  On the other hand, we have
  \begin{align*}
    \frac{\mu_e(1)}{\mu_e(0)} &= \frac{\Pr[X \sim \mu]{X_e=1}}{\Pr[X \sim \mu]{X_e=0}} \\
    &= \frac{\Pr[X \sim \mu]{X_e=1}}{\Pr[X \sim \mu]{u\text{ is not saturated in }X}} \cdot \frac{\Pr[X \sim \mu]{u\text{ is not saturated in }X}}{\Pr[X \sim \mu]{X_e=0}} \\
    &= \lambda \Pr[X \sim \mu_{G-u}]{v\text{ is not saturated in }X} \cdot \Pr[X \sim \mu_{G-e}]{u\text{ is not saturated in }X} \\
    &\ge \frac{\lambda}{(1+\lambda\Delta)^2}
  \end{align*}
  where we use $\mu_{G-u}$(resp.~$\mu_{G-e}$) to denote the monomer-dimer model with edge weight $\lambda$ on graph $G-u$(resp.~$G-e$).
  This implies that $\mu_e(1) \ge \frac{\lambda}{\lambda + (1+\lambda\Delta)^2}$.
\end{proof}
\begin{remark}
  The factor $\frac{1-2\phi}{\log(\frac{1}{\phi} - 1)}=\Omega(\frac{1}{\log\frac{1+\lambda\Delta}{\lambda}})$ in the log-Sobolev bound in \Cref{thm:log-Sob-HS-route-via-congestion} vanishes as $\lambda\to 0^+$. 
  However, when $\lambda=o(\frac{1}{\Delta})$, we can prove an $O(m\log m)$ mixing time for both Glauber dynamics and the Jerrum-Sinclair chain on monomer-dimers via a simple path coupling argument  (see, e.g.,~\cite[Corollary 18]{Dyer2007MatrixNA}).
  As a result, it is sufficient to focus on the case where  $\lambda=\Omega(\frac{1}{\Delta})$. 
  In this regime, the factor is  lower bounded as $\frac{1-2\phi}{\log(\frac{1}{\phi} - 1)}=\Omega(\frac{1}{\log(\ol{\lambda}\Delta)})$, where $\ol{\lambda}\triangleq\max\{1,\lambda\}$.
\end{remark}

\section{Faster Mixing of Monomer-Dimer Glauber Dynamics}\label{sec:faster-mixing-of-monomer-dimer-glauber-dynamics}
\newcommand{\toJ}{\overset{\-{JS}}{\mapsto}}
\newcommand{\toG}{\overset{\-{GD}}{\mapsto}}
\newcommand{\PJ}{P_{\-{JS}}}
\newcommand{\PG}{P_{\-{GD}}}
In this section, we establish an upper bound on the mixing time of Glauber dynamics for the monomer-dimer model via a comparison argument.
\begin{theorem}\label{thm:main-Glauber}
  Let $G=(V,E)$ be a simple undirected graph with $n$ vertices, $m$ edges, and maximum degree~$\Delta$.
  Let $\mu$ be the monomer-dimer distribution on $G$ with edge weight $\lambda>0$, and define $\overline{\lambda} = \max\set{1,\lambda}$.
  The Glauber dynamics for $\mu$,  with transition matrix $\PG$, has Poincar\'e constant $ \gamma(\PG)$ and log-Sobolev constant $\rho(\PG)$ satisfying:
  \begin{align*}
    \gamma(\PG) = \Omega\tp{\frac{1}{\ol{\lambda}\Delta}}\cdot \gamma(P_{\mathrm{JS}}) \quad \text{and} \quad \rho(\PG) = \Omega\tp{\frac{1}{\ol{\lambda}\Delta}}\cdot\rho(P_{\mathrm{JS}}).
  \end{align*}
  As a consequence, the mixing time of the Glauber dynamics for the monomer-dimer model satisfies
  \begin{align*}
      T_{\mathrm{mix}}(\PG) &= O\left(\overline{\lambda}^{4} \Delta^2 m\cdot  \min \left\{n\log^2(1+\ol{\lambda}),\,\, \overline{\lambda} \Delta \log(\ol{\lambda} \Delta)\cdot \log n \right\}\right)=\widetilde{O}_{\lambda}(\Delta^3m).
  \end{align*}
  \end{theorem}
The Markov chain $\PG$ is similar to Jerrum-Sinclair chain $\PJ$,  but it updates at most one edge when transitioning from $X_t$ to $X_{t+1}$ according to the following rules: 
\begin{enumerate}
\item Select an edge $e\in E$ uniformly at random;
\item Set $X_{t+1}(e)$ to a random value sampled from the conditional marginal distribution $\mu^{{X_t}(N(e))}_e,$ where 
  $N(e)=\set{f\in E: f\text{ is incident to }e \text{ in }G}$.
\item  For all  edges $f\in E\setminus\set{e}$, set $X_{t+1}(f)=X_t(f)$.
\end{enumerate}
It can be verified that the Glauber dynamics defined above is ergodic and reversible with respect to the stationary distribution $\mu$. Moreover, $\PG$ is positive semi-definite.

We now prove the following key lemma that compares the Dirichlet forms of the monomer-dimer Glauber dynamics $\PG$ with those of the Jerrum-Sinclair chains $\PJ$.
\begin{lemma}\label{lem:comparing-Dirichlet}
  For any function $f:\Omega \to \mathbb{R}_{\geq 0}$, we have $\+E_{\PJ}(f,f) \le 4\Delta(\ol{\lambda}+1)\cdot \+E_{\PG}(f,f)$.
\end{lemma}
\Cref{thm:main-Glauber} then follows from the straightforward application of \Cref{lem:comparing-Dirichlet} to the definitions of the Poincaré constant and the log-Sobolev constant.
To prove \Cref{lem:comparing-Dirichlet}, it is sufficient to use the standard canonical path method. For more details, we refer the readers to~\cite[Section 13.4]{levin2017markov} for a thorough introduction.


Let $P_1, P_2$ be two ergodic Markov chains on the state space $\Omega$, both reversible with respect to the stationary distribution $\mu$. 
For each transition $x\mapsto y$ in $P_1$ with $P_1(x,y) > 0$, 
we define a canonical path $\gamma^{x,y} = (x_0, x_1, \cdots, x_\ell)$ with length $\ell = \ell(\gamma^{x,y})$ from $x_0=x$ to $x_{\ell}=y$ using the transitions in~$P_2$, which means $P_2(x_i, x_{i+1}) > 0$ for all $i < \ell$.

Given such a family of canonical paths, the \emph{congestion ratio} $B$ between $P_1$ and $P_2$ is defined as
\begin{align*}
  B \triangleq \max_{\alpha, \beta\in\Omega: \atop P_2(\alpha, \beta) > 0} B(\alpha,\beta),\quad \text{ where }\quad B(\alpha,\beta)\triangleq  
  \frac{1}{\mu(\alpha)P_2(\alpha,\beta)} \sum_{x,y\in\Omega: \atop  (\alpha, \beta)\in\gamma^{x,y}} \mu(x) P_1(x,y) \ell(\gamma^{x,y}).
\end{align*}

\begin{lemma}[comparing Dirichlet forms via paths \text{\cite[Theorem 13.20]{levin2017markov}}] \label{lem: comparison-via-congestion}
    If there exists a family of canonical paths with congestion ratio $B$ between $P_1$ and $P_2$, then for all functions $f:\Omega \to \=R$, it holds $$\+E_{P_1}(f,f) \leq B \cdot \+E_{P_2}(f,f).$$ 
\end{lemma}

\begin{proof}[Proof of~\Cref{lem:comparing-Dirichlet}]
To apply \Cref{lem: comparison-via-congestion} with $P_1 = \PJ$ and $P_2 = \PG$,
we provide a  symmetric construction of canonical paths  $\gamma^{x,y}=\gamma^{x,y}$ with congestion ratio $B\le 4\Delta(\ol{\lambda}+1)$ between $P_1$ and~$P_2$. 
Depending on the size of the difference $x\oplus y\triangleq \{e\in E\mid x_e\neq x_e\}$,
the construction proceeds in two cases, ensuring that $\ell(\gamma^{x,y})\le 2$ in both cases.
\begin{enumerate}
  \item When $\abs{x\oplus y}=1$, $x$ and $y$ differ at exactly one edge. Thus, $(x, y)$ is a legal transition in $\PG$.
  We define $\gamma^{x,y}=\tp{x,y}$, which forms a path of length 1 in the transition graph of $\PG$.
  \item When $\abs{x\oplus y}=2$, we have $x_g=0,y_g=1$ and $x_h=1,y_h=0$ for exactly two edges $g,h$, and $x=y$ elsewhere. 
  Then, $(x,z)$ and $(z,y)$ are legal transitions in $\PG$, where $z=x\triangle\{h\}$. We define $\gamma^{x,y}=\tp{x, z, y}$,  which forms a path of length 2 in the transition graph of $\PG$,  
\end{enumerate}
Next, observe the following simple properties of the canonical paths from the construction above.
  Suppose $(\alpha, \beta)\in \gamma^{x,y}$ is an edge-removing transition in Glauber dynamics with $\alpha_e=1$ and $\beta_e=0$. Then:
  \begin{enumerate}
    \item If $\abs{x\oplus y}=1$, then $x=\alpha, y=\beta$.
    \item If $\abs{x\oplus y}=2$ with $x_g=0,y_g=1$, $x_h=1,y_h=0$, we have $e=h$, $x=\alpha$, and $y=\beta\triangle\{g\}$. 
    Note that $g$ is incident to $h$. Since $h=e$ is fixed in the transition $(\alpha,\beta)$, the number of possible choices for $g$ is at most $2(\Delta - 1)$. 
  \end{enumerate}
  For any edge-removing transition $(\alpha, \beta)$ in Glauber dynamics, we bound the congestion ratio:
  \begin{align*}
    B(\alpha, \beta) 
    \le & 2\tp{\frac{\mu(\alpha)\PJ(\alpha,\beta)}{\mu(\alpha)\PG(\alpha,\beta)}+\sum_{g\text{ incident to }h} \frac{\mu(\alpha)\PJ(\alpha,\beta\triangle \{g\})}{\mu(\alpha)\PG(\alpha,\beta)}} \\
    \le & 2\tp{(\lambda+1)\min\tp{1,1/\lambda}+2(\Delta-1)\cdot (\lambda+1) }\\
    \le & 4\Delta(\ol{\lambda}+1).
  \end{align*}
  By symmetry, the same bound also applies to any edge-inserting transition in Glauber dynamics. The conclusion follows by applying \Cref{lem: comparison-via-congestion}.
\end{proof}


\section{Concurrent Work and Future Directions}
We note that a concurrent work \cite{chen2025rapid}, which shares a subset of authors with the current paper, proves an $\widetilde{O}(\sqrt{\Delta}nm)$ mixing time for the Glauber dynamics for uniform sampling of matchings (with edge weight $\lambda=1$), among other results, by utilizing the trickle down method.

An important open question remains: What is the sharp bound for the mixing time of the Jerrum-Sinclair chain on matchings? A potential milestone to aim for is a bound of $\widetilde{O}(\sqrt{\Delta}m)$, where the $\Theta(\sqrt{\Delta})$ factor is motivated by the total influence bound for matchings~\cite{bayati2007simple}.


Beyond the monomer-dimer system, the canonical path method for analyzing mixing times has been applied to a variety of foundational problems, such as the permanent, the Ising model, and the switch/flip chain for generating random regular graphs.
A promising direction for future work is to extend the new techniques developed in this paper to these problems, with the goal of achieving improved mixing times.

As a final remark, 
it will be exciting to see how the key insights of this paper can be applied to overcome some long-standing challenges in current approaches to mixing times, such as the exponential dependence of mixing time on spectral independence.

\pagebreak

\bibliographystyle{alpha}
\bibliography{refs}

\appendix

\pagebreak

\appendix

\section{A Generalization of Our Approach to Localization Schemes}\label{sec:localization-scheme}

In the proof of \Cref{thm:general-results}, random pinnings are introduced to the distribution, progressively transforming it into a Dirac measure.
The concept of \emph{localization schemes}, introduced in \cite{chen2022localization}, generalizes random pinnings to more abstract localization processes.
This abstraction generalizes the local-to-global argument developed in \cite{anari2019logconcave, cryan2021modified}.
Inspired by this, we present a similar generalization of our local-to-global principle, as stated in \Cref{thm:general-results}, using localization schemes.

\paragraph{Noising and Denoising Process}
Let $(X_t)_{t\in [0,1]}$ be a Markov process that takes values in $\Omega$, with $X_0$ drawn from the target distribution $\mu$. The process starts at $X_0$ and progressively removes randomness until the law of $X_1$ becomes a Dirac distribution. We refer to this process $(X_t)_{t\in [0,1]}$ as a \emph{noising process}.
The \emph{denoising process} $(Y_t)_{t\in [0,1]}$ is then defined as the time-reversal of $(X_t)_{t\in [0,1]}$, i.e., $Y_t \triangleq X_{1-t}$.
It is straightforward to verify that $(Y_t)_{t\in [0,1]}$ is also a Markov process.
Since $Y_t \overset{\-d}{=} X_{1-t}$, we know that $Y_1 \sim \mu$ and the law of $Y_0$ is a Dirac distribution.
Next, let $(Q_t)_{t\in [0,1]}$ be a process of transition matrices (row stochastic matrices) in $\=R^{\Omega\times \Omega}$ that adapts to the same filtration as $(Y_t)_{t\in[0,1]}$.
In particular, the (random) stationary distribution of $Q_t$ is the conditional law of $Y_1$ given $Y_t$ (with $Q_0$ having stationary distribution $\mu$). 
We refer to $(Q_t)_{t\in [0,1]}$ as the associated \emph{chain process} with respect to $(Y_t)_{t\in [0,1]}$.

Typical examples of localization schemes include coordinate-by-coordinate localization and negative-field localization. We apply the coordinate-by-coordinate localization from \Cref{thm:general-results} in this framework.

\begin{example}[coordinate-by-coordinate localization]
    Let $\mu$ be a distribution supported over $[q]^n$, and let $\square$ be an auxiliary spin.
    The noising process $(X_t)_{t \in [0,1]}$ with initial distribution $X_0 \sim \mu$ is defined as follows. Initially, we select a total order $r_1 < r_2 < \ldots < r_n$ of all sites $[n]$ uniformly at random. For all $t \in [0,1]$, $X_t$ is obtained by replacing the spin of the first 
$\lfloor nt\rfloor$ sites, $r_1,r_2,\ldots,r_{\lfloor nt\rfloor}$ in $X_0$ with $\square$.
\end{example}

\begin{example}[negative-field localization]
Let $\mu$ be a distribution supported over $2^{[n]}$. The noising process $(X_t)_{t \in [0,1]}$ with initial distribution $X_0 \sim \mu$ is defined as follows. Initially, we assign each site $i \in [n]$ with an independent random variable $r_i$ uniformly distributed on $[0,1]$. For all $t \in [0,1]$, $X_t$ is the set of sites $i \in X_0$ with $r_i \ge t$.
\end{example}

We can generalize the local functional inequalities defined in~\Cref{def:HS-route}.
\begin{definition}[local functional inequalities]
  Let $\mu$ be a distribution supported on $\Omega$ and let $(Y_t)_{t\in [0,1]}$ be a denoising process with $Y_1 \sim \mu$ and let $(Q_t)_{t\in [0,1]}$ be the chain process with respect to $(Y_t)_{t\in [0,1]}$. Let $\alpha:[0,1] \to \=R_{\geq 0}$ be a function.
  \begin{itemize}
  \item (local Poincar\'e inequality) We say that $((Y_t),(Q_t))$ satisfies the \emph{$\alpha$-local Poincar\'e inequality} if 
  \begin{align} \label{eq:def-local-poincare-general}
    \alpha(t) \cdot \lim_{h\to 0^+} \frac{1}{h} \Var{\E{f(Y_1) \mid Y_{t+h}} \mid Y_t} \leq  \+E_{Q_t}(f,f),\quad \forall f: \Omega(\mu) \to \mathbb{R}.
  \end{align}
  \item (local log-Sobolev inequality) We say $((Y_t),(Q_t))$ satisfies the \emph{$\alpha$-local log-Sobolev inequality} if 
  \begin{align} \label{eq:def-local-log-sobolev-general}
    \alpha(t) \cdot \lim_{h\to 0^+} \frac{1}{h} \Ent{\E{f^2(Y_1) \mid Y_{t+h}} \mid Y_t} \leq  \+E_{Q_t}(f,f),\quad \forall f: \Omega(\mu) \to \mathbb{R}.
  \end{align}
  \end{itemize}
\end{definition}
\begin{remark}
    For coordinate-by-coordinate localization, the variance $\Var{\E{f(Y_1) \mid Y_{t+h}} \mid Y_t}$ or entropy $\Ent{\E{f^2(Y_1) \mid Y_{t+h}} \mid Y_t}$ is not differentiable at $t = \frac{k}{n}$ for all $1 \le k \le n$. In this case, we use the finite difference instead, which is consistent with our previous definitions (\Cref{def:HS-route}).
\end{remark}

We can also generalize the notion of concave Dirichlet forms in terms of $(Q_t)_{t\in [0,1]}$.
\begin{definition}[concave Dirichlet forms]\label{def:concave-Dirichlet-forms-new}
  We say the process $(Q_t)_{t\in [0,1]}$ have has \emph{concave Dirichlet forms} if for any $t \in [0,1]$ and any function $f:\Omega \to \mathbb{R}$,
  \begin{align*}
    \E{\+E_{Q_t}(f,f)} \leq \+E_{Q_0}(f,f).
  \end{align*}
\end{definition}

The following presents a generalization of the ``local-to-global'' theorem (\Cref{thm:general-results}) for functional inequalities, extending it to the general context of localization schemes.
\begin{theorem}
  Let $\mu$ be a distribution over $\Omega$.
  Let $(Y_t)_{t\in [0,1]}$ be a denoising process for $\mu$ and let $(Q_t)_{t\in [0,1]}$ be the chain process with respect to $(Y_t)_{t\in [0,1]}$ such that $(Q_t)_{t\in [0,1]}$ has {concave Dirichlet forms}. Let $\alpha:[0,1] \to \=R_{\geq 0}$ be a function.
  \begin{itemize}
  \item If $((Y_t),(Q_t))$ satisfies the $\alpha$-local Poincar\'e inequality, then the Poincar\'e constant $\gamma(Q_0)$ satisfies
    \begin{align*}
      \gamma(Q_0) \geq \tp{\int_0^1 \alpha^{-1}(t) \;\-d t}^{-1}.
    \end{align*}
  \item If $((Y_t),(Q_t))$ satisfies the $\alpha$-local log-Sobolev inequality, then the log-Sobolev constant $\rho(Q_0)$ satisfies
    \begin{align*}
      \rho(Q_0) \geq \tp{\int_0^1 \alpha^{-1}(t) \;\-d t}^{-1}.
    \end{align*}
  \end{itemize}
\end{theorem}
\begin{proof}
  We give the proof for the Poincar\'e constant $\gamma(Q_0)$, the log-Sobolev constant follows by almost the same proof.
  Let $\oVar(t) \triangleq \E{\Var{f(Y_1) \mid Y_t}}$ as a function of $t \in [0,1]$.
  Hence, the derivative of this function can be bounded by
  \begin{align*}
    -\oVar(t)'
    &= \lim_{h\to 0^+} \frac{1}{h} \tp{\oVar(t) - \oVar(t+h)} \\
    &=  \lim_{h\to 0^+} \frac{1}{h} \E{ \Var{f(Y_1) \mid Y_t} - \E{\Var{f(Y_1) \mid Y_{t+h}} \mid Y_t} } \tag{law of total expectation} \\
    &=  \lim_{h\to 0^+} \frac{1}{h} \E{ \Var{\E{f(Y_1) \mid Y_{t+h}} \mid Y_t} } \tag{law of total variance} \\
    &=  \E{\lim_{h\to 0^+} \frac{1}{h} \Var{\E{f(Y_1) \mid Y_{t+h}} \mid Y_t} } \\
    &\leq  \alpha^{-1}(t) \E{\+E_{Q_t}(f,f)} \tag{$\alpha$-local Poincar\'e inequality} \\
    &\leq  \alpha^{-1}(t) \+E_{Q_0}(f,f) \tag{concave Dirichlet forms}.
  \end{align*}
  We finish the proof by taking an integration at both side and get
  \begin{align*}
    \Var{f(Y_1)} &= \oVar(0) - \oVar(1) = - \int_0^1 \oVar(t)' \-d t \leq \tp{\int_0^1 \alpha^{-1}(t) \-d t} \+E_{Q_0}(f,f). \qedhere
  \end{align*}
\end{proof}

\end{document}